\documentclass[smallabstract,smallcaptions,12pt]{dccpaper}
\usepackage{fullpage}
\usepackage{microtype}

\usepackage{enumitem}

\usepackage[
  algo2e,
  vlined,
  linesnumbered,
  noend,
  algosection,
  english]{algorithm2e}
\DontPrintSemicolon
\SetAlCapHSkip{0pt}

\SetCommentSty{ThesisCommentStyle}
\SetAlgoCaptionSeparator{.} %
\SetFuncSty{}
\SetArgSty{}

\SetKwProg{Fn}{Function}{}{end}\SetKwFunction{FRecurs}{FnRecursive}%
\SetAlgoLongEnd

\newcommand{\nosemic}{\renewcommand{\@endalgocfline}{\relax}}%
\newcommand{\dosemic}{\renewcommand{\@endalgocfline}{\algocf@endline}}%
\let\oldnl\nl%
\newcommand{\nonl}{\renewcommand{\nl}{\let\nl\oldnl}}%

\newcommand{\Remx}[1]{\;}
\newcommand{\Remxi}[1]{}

\usepackage{hyperref}
\usepackage[disable]{todonotes}
\usepackage{wrapfig}
\usepackage{tabularx}
\usepackage{booktabs}
\usepackage{multicol}
\usepackage{multirow}
\usepackage{subcaption}
\usepackage{graphicx}
\usepackage{placeins}

\usepackage{amsmath}
\usepackage{amssymb}
\usepackage{framed}
\usepackage[hang,flushmargin]{footmisc}

\usepackage{amsthm}
\newtheorem{theorem}{Theorem}[section]

\newtheorem{lemma}[theorem]{Lemma}
\newtheorem{definition}{Definition}[section]

\usepackage[capitalise,noabbrev]{cleveref}
\crefname{line}{line}{lines}

\newcommand{\ceil}[1]{\ensuremath\lceil#1\rceil}
\newcommand{\floor}[1]{\ensuremath\lfloor#1\rfloor}
\newcommand{\Ohtilde}[1]{\ensuremath \tilde{\mathrm{O}}(#1)}

\newcommand{\access}{{access}}
\newcommand{\rank}{{rank}}
\newcommand{\select}{{select}}

\newcommand{\error}{\ensuremath{\epsilon}}
\newcommand{\arank}{\ensuremath{{\rank}^\approx}}

\newcommand{\qvector}[1]{QV\(_{#1}\)}

\newcommand{\bvsdsl}{sdsl\_bv}
\newcommand{\bvpasta}{pasta\_bv}

\newcommand{\wtqvector}[1]{QWM\(_{#1}\)}
\newcommand{\wtqvectorpfs}[1]{QWM\(_{#1}^{\textnormal{pfs}}\)}
\newcommand{\wtsdsl}{sdsl\_wm}
\newcommand{\wtpasta}{pasta\_wm}

\newcommand{\english}{{English}}
\newcommand{\dna}{{DNA}}
\newcommand{\cc}{{CC}}
\newcommand{\wiki}{{Wiki}}

\newcommand{\mytodo}[1]{\textcolor{red}{\textbf{TODO}: #1}}
\renewcommand{\mytodo}[1]{}

\usepackage{tikz,pgfplots}
\pgfplotsset{compat=1.15}

\usepackage{xcolor}
\definecolor{my-dark-red}{RGB}{183, 28, 28}
\definecolor{my-red}{RGB}{244,67,54}
\definecolor{my-pink}{RGB}{233,30,99}
\definecolor{my-purple}{RGB}{156,39,176}
\definecolor{my-deep-purple}{RGB}{103,58,183}
\definecolor{my-indigo}{RGB}{63,81,181}
\definecolor{my-blue}{RGB}{33,150,243}
\definecolor{my-light-blue}{RGB}{3,169,244}
\definecolor{my-cyan}{RGB}{0,188,212}
\definecolor{my-teal}{RGB}{0,150,136}
\definecolor{my-green}{RGB}{76,175,80}
\definecolor{my-light-green}{RGB}{139,195,74}
\definecolor{my-lime}{RGB}{205,220,57}
\definecolor{my-yellow}{RGB}{255,235,59}
\definecolor{my-amber}{RGB}{255,193,7}
\definecolor{my-orange}{RGB}{255,152,0}
\definecolor{my-deep-orange}{RGB}{255,87,34}
\definecolor{my-brown}{RGB}{121,85,72}
\definecolor{my-grey}{RGB}{158,158,158}
\definecolor{my-blue-grey}{RGB}{96,125,139}
\definecolor{my-lipics-grey}{rgb}{0.6,0.6,0.61}

\definecolor{colorPASTA}{HTML}{444444}
\definecolor{colorSDSL}{HTML}{377EB8}
\definecolor{colorQWM256}{HTML}{A65628}
\definecolor{colorQWM512}{HTML}{4DAF4A}

\pgfplotsset{
  mark repeat*/.style={
    scatter,
    scatter src=x,
    scatter/@pre marker code/.code={
      \pgfmathtruncatemacro\usemark{
        or(mod(\coordindex,#1)==0, (\coordindex==(\numcoords-1))
      }
      \ifnum\usemark=0
        \pgfplotsset{mark=none}
      \fi
    },
    scatter/@post marker code/.code={}
  },
  major grid style={thin,dotted},
  minor grid style={thin,dotted},
  ymajorgrids,
  yminorgrids,
  every axis/.append style={
    line width=0.7pt,
    tick style={
      line cap=round,
      thin,
      major tick length=4pt,
      minor tick length=2pt,
    },
    mark options={solid},
  },
  legend cell align=left,
  legend style={
    line width=0.7pt,
    /tikz/every even column/.append style={column sep=3mm,black},
    /tikz/every odd column/.append style={black},
    mark options={solid},
  },
  legend style={font=\small},
  title style={yshift=-2pt},
  enlarge x limits=0.04,
  every tick label/.append style={font=\footnotesize},
  every axis label/.append style={font=\small},
  every axis y label/.append style={yshift=-1ex},
  /pgf/number format/1000 sep={},
  axis lines*=left,
  xlabel near ticks,
  ylabel near ticks,
  axis lines*=left,
  label style={font=\footnotesize},
  tick label style={font=\footnotesize},
  plotLatency/.style={
    width=49.5mm,
    height=57.5mm,
  },
  plotLatencySmallWide/.style={
    width=45.5mm,
    height=57.5mm,
  },
  plotLatencySmall/.style={
    width=42.5mm,
    height=45.5mm,
  },
}

\newcommand{\longversion}[1]{#1}

\newcommand{\longandshortversion}[2]{#1}

\newlength{\figurewidth}
\newlength{\smallfigurewidth}

\begin{document}

\title
{\large
\textbf{Faster Wavelet Tree Queries}
}
\date{}
\author{%
Matteo Ceregini$^{\ast}$, Florian Kurpicz$^{\dag}$, and Rossano Venturini$^{\ast}$\\[0.5em]
{\small\begin{minipage}{\linewidth}\begin{center}
\begin{tabular}{ccc}
$^{\ast}$University of Pisa & \hspace*{0.5in} & $^{\dag}$Karlsruhe Institute of Technology \\
\url{ceregini@studenti.unipi.it} && \url{kurpicz@kit.edu}\\
\url{rossano.venturini@unipi.it} && 
\end{tabular}
\end{center}\end{minipage}}
}

\maketitle
\thispagestyle{empty}

\begin{abstract}
  \noindent
  Given a text, rank and select queries return the number of occurrences of a character up to a position (rank) or the position of a character with a given rank (select).
  These queries have applications in, e.g., compression, computational geometry, and most notably pattern matching in the form of the backward search---the backbone of many compressed full-text indices.
  Currently, in practice, for text over non-binary alphabets, the wavelet tree is probably the most used data structure for rank and select queries.\longandshortversion{

  In this paper, we present techniques to speed up queries by a factor of two (access and select) up to three (rank), compared to the wavelet tree implementation contained in the widely used Succinct Data Structure Library (SDSL).
  To this end, we change the underlying tree structure from a binary tree to a 4-ary tree and reduce cache misses by approximating rank queries using a predictive model to prefetch all data required for the actual rank query.}{ Our improved wavelet tree representation and predictive model allows us to speed up queries by a factor of 2--3.}
\end{abstract}

\section{Introduction}
Wavelet trees~\cite{GrossiGV2003WaveletTree} are a compressible self-indexing rank and select data structure, i.e., they can answer rank (number of occurrences of symbol up to position \(i\)) and select (position of \(i\)-th occurrence of symbol) queries, while still allowing to access the text.
This makes them an important building block for compressed full-text indices, e.g., the FM-index~\cite{FerraginaM2000FMIndex} or the r-index~\cite{GagieNP2020FullyFunctionalRIndex}, where they are used to answer rank queries during the pattern matching algorithm---the backwards search\longversion{---cf. ~\cref{ex:backward_search}}.\longandshortversion{

Wavelet trees have many applications, which are discussed in multiple surveys~\cite{FerraginaGM2009MyriadWT,GrossiVX2011WaveletSurvey,Makris2012WaveletSurvey,Navarro2014WaveletForAll}.
Due to the plethora of applications, a lot of research has been focused on the efficient construction of wavelet trees in both practice and theory.
We give an overview of the state-of-the-art in \cref{sec:related_work}.}{ More applications are discussed in multiple surveys~\cite{FerraginaGM2009MyriadWT,GrossiVX2011WaveletSurvey,Makris2012WaveletSurvey,Navarro2014WaveletForAll}. Due to the plethora of applications, a lot of research has been focused on the efficient construction of wavelet trees.}
However, there exists barely any research focusing on the query performance of wavelet trees.
While there exist alternative representations of the wavelet tree (namely the wavelet matrix\longversion{, see \cref{sec:preliminaries}}) that provide better practical query performance, the better query performance is more of a byproduct of a space efficient representation for large alphabets.\longandshortversion{

The main building block of wavelet trees (and wavelet matrices) are bit vectors with binary rank and select support.
There exist many different approaches tuning the rank and select support for query time and/or space overhead.
Faster binary rank and select queries directly translate to faster queries on wavelet trees.
We refer to \cref{sec:related_work} for an overview of binary rank and select support for bit vectors.
However, improving only the binary rank and select data structure still not fully utilizes the full range of optimizations when it comes to answering queries using wavelet trees.}{ The main building block of wavelet trees are bit vectors with binary rank and select support. There exist many different approaches tuning the rank and select support for query time and/or space overhead. Faster binary rank and select queries directly translate to faster queries on wavelet trees. However, improving only the binary rank and select data structure still not fully utilizes the full range of optimizations.}

\longversion{
\begin{algorithm2e}[b]
  \Fn{BackwardsSearch{$(P[1..m],C,wt)$}}{
    
    $s=1, e=n$\;
    \For{$i=m,\dots,1$}{
      \(s=C[P[i]]+wt.rank_{P[i]}(s-1)+1\)\;
      \(e=C[P[i]]+wt.rank_{P[i]}(e)\)\;
      \If{\(s>e\)}{
        \Return \(\emptyset\)\;
      }
    }
    \Return \([s,e]\)\;
  }
  \vspace{.25cm}
\caption{Backward search for a pattern \(P\) of length \(m\) using the wavelet tree \(wt\) over the Burrows-Wheeler transform of the text and the exclusive prefix sum over the histogram of characters \(C\).\label{ex:backward_search}}
\end{algorithm2e}
}

\longandshortversion{\subsection*{Our Contributions}}{\paragraph{Our Contributions.}}
\longandshortversion{Wavelet trees usually utilize binary trees as underlying tree structure.
We show that using a 4-ary tree as underlying tree structure (see \cref{sec:four_ary_wavelet_trees}) results in a query speedup of up to 2 for all queries compared to its competitor implemented in the widely used \emph{Succinct Data Structure Library} (SDSL)~\cite{GogBMP2014SDSL}.}{We show that using a 4-ary wavelet tree instead of the usual binary wavelet tree results in a query speedup of up to 2 for all queries compared to its competitor implemented in the widely used \emph{Succinct Data Structure Library} (SDSL)~\cite{GogBMP2014SDSL}.}
Furthermore, we introduce the \emph{rank with additive approximation} problem (see \cref{sec:prefetching}) and show how utilize a small prediction model to locate data necessary during rank queries.
We use this information to improve rank queries (which are required for pattern matching) even more, achieving a total speedup of up to 3, by prefetching all data necessary to answer the query\longandshortversion{}{, see \cref{sec:exeperimental_evaluation}}. \longversion{

In our experimental evaluation (see \cref{sec:exeperimental_evaluation}), we not only show these impressive speedups for such a well-researched data structure but also that our data structure requires less space and is faster to construct---making it strictly superior to its competitors.}

\longandshortversion{\section{Preliminaries}}{\paragraph{Preliminaries.}}
\label{sec:preliminaries}
A \emph{bit vector} is a text over the alphabet \(\{0,1\}\).
Given a text \(T\) of length \(n\) over an alphabet \(\Sigma=[0,\sigma)\).
For \(i\in[0,n)\) and \(\alpha\in\Sigma\), we want to answer:
\longandshortversion{
  
\begin{itemize}
\item \(rank_\alpha(i)=|\{j<i\colon T[j]=\alpha\}|\) and
\item \(select_\alpha(i)=\min\{j\colon rank_\alpha(j)=i\}\).
\end{itemize}
}{
\[rank_\alpha(i)=|\{j<i\colon T[j]=\alpha\}|\textnormal{~and~}select_\alpha(i)=\min\{j\colon rank_\alpha(j)=i\}.\]
}

\longandshortversion{Rank and select}{Both} queries on bit vectors of length \(n\) can be answered in \(O(1)\) time with \(o(n)\) additional bits \cite{ClarkM1996Select,Jacobson1989LOUDS}.
The \emph{most significant bit} (MSB) of a character is the bit with the highest value.
\longandshortversion{For simplicity, we}{We} assume that the MSB is the leftmost bit.
The \(i\)-th MSB is the bit with the \(i\)-th highest value.
A length-\(\ell\) \emph{bit-prefix} of a character are the \longandshortversion{character's}{its} \(\ell\) MSBs.

A \emph{wavelet tree} \cite{GrossiGV2003WaveletTree} is a binary tree, where each node represents a subsequence of the text.
Each node contains character with a specific length-\(k\) bit-prefix.
The root of a wavelet tree represents all characters with the length-0 bit empty prefix, i.e., all characters.
Then, whenever we visit a left child of a node that represents characters with bit-prefix \(\alpha\), the child represents character with-bit prefix \(\alpha\texttt{0}\).
The right child represents characters with bit-prefix \(\alpha\texttt{1}\).
\longversion{Alternatively, you can say that the left child represents characters that are in the lower half of the alphabet represented in its parent and the right child represents characters in the upper half.
Here, the root represents characters in the whole alphabet.

} On the \(\ell\)-th level of the tree (the root has level \(1\)), characters are represented by their \(\ell\)-th MSB.
Within a node, all represented characters are stored in a bit vector.
If we concatenate the bit vectors of all nodes on the same level, we obtain a \emph{level-wise} wavelet tree.
We say that all characters that have been represented in a node of a non-level-wise wavelet tree are in the same interval.
See \cref{fig:example_wt_wm} for an example.
\longversion{All intervals in a wavelet tree can be identified by the bit prefix of the characters represented by that interval.}
In the following, we consider a level-wise wavelet trees.

\begin{figure*}[t]
  \centering
\begin{tikzpicture}[
    parent anchor=south,
    child anchor=north,
    every node/.style={
        font=\scriptsize,
        shape=rectangle,
        draw,
        align=center,
        minimum height=1.75em,
    },
    level 1/.style={
        level distance=6ex,
        sibling distance=8em,
    },
    level 2/.style={
        level distance=6ex,
        sibling distance=4em,
    },
    level 3/.style={
        level distance=6ex,
        sibling distance=1em,
    },
]
\node(root){\textcolor{gray}{\texttt{accessandselect}}\\\texttt{000011010101001}}
    child {
        node (v1) { \textcolor{gray}{\texttt{acceadeec}}\\\texttt{000101110} }
        child { node (v3) { \textcolor{gray}{\texttt{accac}}\\\texttt{01101} } }
        child { node (v4) { \textcolor{gray}{\texttt{edee}}\\\texttt{1011} } }
    }
    child {
        node (v2) { \textcolor{gray}{\texttt{ssnslt}}\\\texttt{110101} }
        child { node (v5) { \textcolor{gray}{\texttt{nl}}\\\texttt{10} } }
        child { node (v6) { \textcolor{gray}{\texttt{ssst}}\\\texttt{0001} } }
    };

\draw[latex-latex] ([yshift=-1.25ex]v1.east) -- ([yshift=-1.25ex]v2.west);
\draw[latex-latex] ([yshift=-1.25ex]v3.east) -- ([yshift=-1.25ex]v4.west);
\draw[latex-latex] ([yshift=-1.25ex]v4.east) -- ([yshift=-1.25ex]v5.west);
\draw[latex-latex] ([yshift=-1.25ex]v5.east) -- ([yshift=-1.25ex]v6.west);

\end{tikzpicture}
\hspace{2cm}
\begin{tikzpicture}[
    parent anchor=south,
    child anchor=north,
    every node/.style={
        font=\scriptsize,
        shape=rectangle,
        draw,
        align=center,
        minimum height=1.75em,
    },
    level 1/.style={
        level distance=6ex,
        sibling distance=8em,
    },
    level 2/.style={
        level distance=6ex,
        sibling distance=4em,
    },
    level 3/.style={
        level distance=6ex,
        sibling distance=1em,
    },
]
\node(root){\textcolor{gray}{\texttt{accessandselect}}\\\texttt{000011010101001}}
    child {
        node (v1) { \textcolor{gray}{\texttt{acceadeecssnslt}}\\\texttt{000101110110101} }
        child { node (v3) { \textcolor{gray}{\texttt{accacnledeessst}}\\\texttt{011011010110001} } }
    };

\node[draw=none,node distance=2cm,right of=root] { \(\textsf{Z}[0]=9\) };
\node[draw=none,node distance=2cm,right of=v1] { \(\textsf{Z}[1]=7\) };
\node[draw=none,node distance=2cm,right of=v3] { \(\textsf{Z}[2]=7\) };

\end{tikzpicture}

\caption{Wavelet tree (left) and a wavelet matrix (right) for the text \texttt{accessandselect} over the alphabet \(\{\texttt{a}~\textcolor{gray}{(\texttt{000})_2},\texttt{c}~\textcolor{gray}{(\texttt{001})_2},\texttt{d}~\textcolor{gray}{(\texttt{010})_2},\texttt{e}~\textcolor{gray}{(\texttt{011})_2},\texttt{l}~\textcolor{gray}{(\texttt{100})_2},\texttt{n}~\textcolor{gray}{(\texttt{101})_2},\texttt{s}~\textcolor{gray}{(\texttt{110})_2},\texttt{t}~\textcolor{gray}{(\texttt{111})_2}\}\) (bit representation of characters given in gray).
  Note that we depict the text for better readability only; the text is not part of the wavelet tree or wavelet matrix.
\longversion{By concatenating all bit vectors in nodes of the wavelet tree that are connected by arrows, we obtain a level-wise wavelet tree.
On each level in the wavelet matrix, there are the same intervals as in the wavelet tree.}}
  \label{fig:example_wt_wm}
\end{figure*}
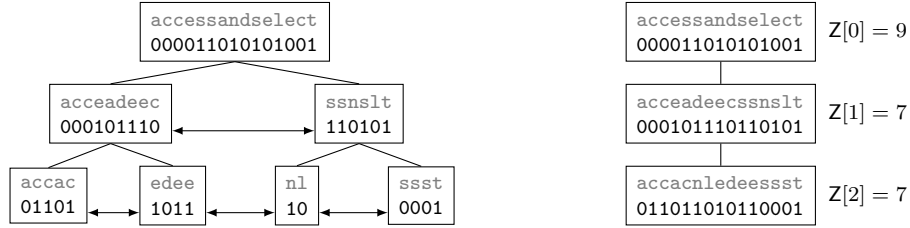

The \emph{wavelet matrix}~\cite{ClaudeNP2015WaveletMatrix} is an alternative representation of the wavelet tree.
The first level of the wavelet matrix are the MSBs of the characters, the same as the first level of the wavelet tree.
Then, to compute the next level \(\ell\), starting with the second, the text is stably sorted using the \((\ell-1)\)-th MSB as key.
Just as with the wavelet tree, the characters are represented using their \(\ell\)-th MSB on each level \(\ell\).
The order of the characters on each level is given by the stably sorted text.
Sorting the text looses the tree structure of the wavelet tree.
However, the same intervals as in the wavelet tree occur on each level, just in a bit-reversal permutation\footnote{See \url{https://oeis.org/A030109}, last accessed 2023-11-08.} order.
A comparison of the structure of a wavelet tree and a wavelet matrix can be found in \cref{fig:example_wt_wm}.
The number of zero in each level is stored in the array \(\textsf{Z}\), which are needed to answer queries using one less binary rank and/or select query per level compared to wavelet trees.
\longversion{We give the rank query algorithms for wavelet trees and matrices in \cref{ex:rank_wavelet_tree,ex:rank_wavelet_matrix}, resp.}
In the following, we use wavelet tree to refer to both wavelet tree and wavelet matrix.
\longversion{A wavelet trees for a text over an alphabet of size \(\sigma\) can answer access, rank, and select queries in \(O(\log\sigma)\) time.}

\longversion{
\begin{algorithm2e}[h!]
  \Fn{WaveletTreeRank{\(_\alpha(i)\)}}{

    \(start=0, size=n, bit\_mask=1\ll (levels-1)\)\;
    \For{\(level=0,\dots,\ceil{\log\sigma}\) and \(i>0\)}{
      \(before=bv.rank_1(start)\)\;
      \(position=bv.rank_1(start+i)-before\)\;
      \(in=bv.rank_1(start+size)-before\)\;
      \uIf{\(\alpha\&bit\_mask\)}{
        \(start = start + (size - in)\)\;
        \(size = in\)\;
        \(i=position\)\;
      }
      \uElse{
        \(size = size - in\)\;
        \(i = i - position\)\;
      }
      \(start+=n\)\;
      \(bit\_mask = bit\_mask\gg 1\)
    }
    \Return \(i\)\;
  }
  \vspace{.25cm}
  \caption{Rank query for a \emph{wavelet tree} over a text of length \(n\) over an alphabet of size \(\sigma\) where all levels are stored in one consecutive bit vector \(bv\).\label{ex:rank_wavelet_tree}}
\end{algorithm2e}

\begin{algorithm2e}[h!]
  \Fn{WaveletMatrixRank{\(_\alpha(i)\)}}{

    \(start=0, size=n, bit\_mask=1\ll (levels-1)\)\;
    \For{\(level=0,\dots,\ceil{\log\sigma}\) and \(i>0\)}{
      \(before=bv.rank_1(start)\)\;
      \(position=bv.rank_1(start+i)-before\)\;
      \(in=before - Z_1[level]\)\;
      \uIf{\(\alpha\&bit\_mask\)}{
        \(i=before\)\;
        \(start=((level + 1) n) + Z_0[level] + in\)\;
      }
      \uElse{
        \(i= i - before\)\;
        \(start=((level + 1)n)+(start - (level\cdot n)) - in \)\;
      }
      \(start+=n\)\;
      \(bit\_mask = bit\_mask\gg 1\)
    }
    \Return \(i\)\;
  }
  \vspace{.25cm}
  \caption{Rank query for a \emph{wavelet matrix} over a text of length \(n\) over an alphabet of size \(\sigma\) where all levels are stored in one consecutive bit vector \(bv\).
    Additionally, \(Z_0[\ell]\) denotes the zeros on level \(\ell\) and \(Z_1[\ell]\) denotes the number of ones before level \(\ell\) in \(bv\).\label{ex:rank_wavelet_matrix}}
\end{algorithm2e}
}

\longandshortversion{\section{Related Work}}{\paragraph{Related Work.}}
\label{sec:related_work}
\longversion{The compact and compressed representation of texts with support for access, rank, and select queries (among others) is an active field of research.
For example, bit vectors \longversion{with rank and select support, i.e., binary rank and select data structures, }are often a building block for succinct data structures.}

\longversion{\subsection*{Binary Rank and Select Data Structures}}
\longversion{For bit vectors of length \(n\), rank and select data structures with constant query time can be constructed in linear time requiring \(o(n)\) space~\cite{ClarkM1996Select,Jacobson1989LOUDS}.}
Practical and well-performing implementations of \longandshortversion{these data}{rank and select} structures can be found in the SDSL~\cite{GogBMP2014SDSL}.
The currently most space efficient rank and select support for a size-\(u\) bit vector that contains \(n\) ones requires only \(\log\binom{u}{n}+\frac{u}{\log u}+\Ohtilde{u^{\frac{3}{4}}}\) bits (including the bit vector)~\cite{Patrascu2008Succincter}.
In practice, the currently fastest select data structures are by Vigna~\cite{Vigna2008BroadwordRankSelect}.
\longversion{Allowing for multiple configurations using a tuning parameter, they outperform all other select data structures while being space-efficient.}
However, they still require much more space than the currently most space-efficient data structures\longandshortversion{ by Zhou~et~al.~\cite{ZhouAK2013PopcountRankSelect} that have recently been improved w.r.t. query throughput by Kurpicz~\cite{Kurpicz2022RankSelect}}{~\cite{Kurpicz2022RankSelect,ZhouAK2013PopcountRankSelect}}.
\longversion{There exist many more practical rank and select data structures that are outperformed by the ones mentioned above~\cite{GonzalezGMN2005PracticalRankSelect,KimNKP2005RankAndSelect,NavarroP2012CombinedSampling}.
Another line of research considers compressed~\cite{ArroyueloW2020CompressedRankSelect,BeskersF2014CompressedRankSelect,BoffaFV2022LearnedRankSelect,KarkkainenKP2014HybBitVector,RamanRS2007RRR} and mutable~\cite{PibiriK2021MutableRankSelect,Prezza2017DYNAMIC} bit vectors with rank and select support.}

\longversion{\subsection*{Wavelet Trees and Wavelet Matrices}}
\longversion{Wavelet tree construction is a well studied field.}
Let \(T\) be a text of length \(n\) over an alphabet of size \(\sigma\).
The\longversion{ asymptotically} best sequential wavelet tree construction algorithms require \(O(n\log\sigma/\sqrt{\log n})\) time~\cite{BabenkoGKS2015WT,MunroNV2016WT}.
These approaches make use of vectorized instructions\longversion{, i.e., SIMD (single instruction, multiple data), to achieve their running time}.
There also exist implementations that make use of these instructions which are available in modern CPUs~\cite{DinklageFKT2023SIMDWT,Kaneta2018VectorizedWXConstruction} and are reported to be the fastest in practice.
In shared memory, wavelet trees can be computed in \(O(\sigma+\log n)\) time requiring only \(O(n\log\sigma/\sqrt{\log n})\) work~\cite{Shun2020ImprParWTandRankSelect}.
In practice, the fastest construction algorithms are based on domain decomposition~\cite{LabeitSB2017ParallelWXSACA,SepulvedaEFS2017DomainDecomposition}\longandshortversion{, where partial wavelet trees are computed in parallel and are then merged also in parallel, using a bottom-up construction for the partial wavelet tree construction~\cite{DinklageEFKL2021PracticalWaveletTrees}}{ and utilize a bottom-up construction as sequential base-case~\cite{DinklageEFKL2021PracticalWaveletTrees}}.
\longversion{Wavelet trees can also be computed in other models of computation, e.g., distributed~\cite{DinklageFK2020DistributedWX} and external memory~\cite{EllertK2019ExternalWX}.}
To compress a wavelet tree, it is constructed for the Huffman-compressed text.\longversion{\footnote{Bit-wise negated canonical Huffman codes are required~\cite{DinklageEFKL2021PracticalWaveletTrees}.}}
The bit vectors in the Huffman-shaped wavelet tree requires \(n\ceil{H_0(T)}\) bits of space, where \(H_0\) is the zeroth order entropy of the text.
\longandshortversion{A fully functional wavelet tree requires binary rank and select support on the bit vectors and needs \(n\ceil{\log\sigma}(1+o(1))\) bits of space \((n\ceil{H_0(T)}(1+o(1)\) bits of space for the Huffman-shaped wavelet tree).}{
  A fully functional Huffman-shaped wavelet tree \((n\ceil{H_0(T)}(1+o(1)\) bits of space.}
\longversion{There are also wavelet trees for degenerate strings~\cite{AlankoBPV2023SubsetWaveletTrees}.}
In theoretical work, multi-ary wavelet trees have been considered before with the main goal to reduce query time in the RAM model to $\Theta(\log_{\log n} \sigma)$~\cite{FerraginaMMN2007MultiAryWaveletTree}.

\longversion{
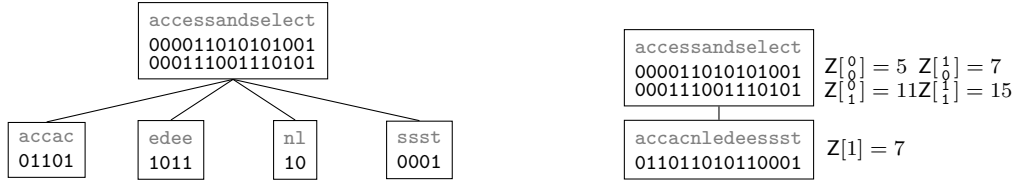
\begin{figure*}[t]
  \centering
\begin{tikzpicture}[
    parent anchor=south,
    child anchor=north,
    every node/.style={
        font=\scriptsize,
        shape=rectangle,
        draw,
        align=center,
        minimum height=1.75em,
    },
    level 1/.style={
        level distance=8ex,
        sibling distance=4em,
    },
    level 2/.style={
        level distance=6ex,
        sibling distance=4em,
    },
    level 3/.style={
        level distance=6ex,
        sibling distance=1em,
    },
]
\node(root){\textcolor{gray}{\texttt{accessandselect}}\\\texttt{000011010101001}\\[-.1cm]\texttt{000111001110101}}
  child { node (v3) { \textcolor{gray}{\texttt{accac}}\\\texttt{01101} } }
  child { node (v4) { \textcolor{gray}{\texttt{edee}}\\\texttt{1011} } }
  child { node (v5) { \textcolor{gray}{\texttt{nl}}\\\texttt{10} } }
  child { node (v6) { \textcolor{gray}{\texttt{ssst}}\\\texttt{0001} } };

\end{tikzpicture}
\hspace{2cm}
\begin{tikzpicture}[
    parent anchor=south,
    child anchor=north,
    every node/.style={
        font=\scriptsize,
        shape=rectangle,
        draw,
        align=center,
        minimum height=1.75em,
    },
    level 1/.style={
        level distance=6ex,
        sibling distance=8em,
    },
    level 2/.style={
        level distance=6ex,
        sibling distance=4em,
    },
    level 3/.style={
        level distance=6ex,
        sibling distance=1em,
    },
]
\node(root){\textcolor{gray}{\texttt{accessandselect}}\\\texttt{000011010101001}\\[-.1cm]\texttt{000111001110101}}
    child {
        node (v1) { \textcolor{gray}{\texttt{accacnledeessst}}\\\texttt{011011010110001} }
    };

\node[draw=none,node distance=2cm,right of=root] (z1) { \(\textsf{Z}[\begin{smallmatrix} \texttt{0} \\ \texttt{0} \end{smallmatrix}]=5\) };
\node[draw=none,node distance=1.25cm,right of=z1] (z2) { \(\textsf{Z}[\begin{smallmatrix} \texttt{1} \\ \texttt{0} \end{smallmatrix}]=7\) };
\node[draw=none,node distance=2.07125cm,yshift=-.325cm,right of=root] (z3) { \(\textsf{Z}[\begin{smallmatrix} \texttt{0} \\ \texttt{1} \end{smallmatrix}]=11\) };
\node[draw=none,node distance=1.25cm,right of=z3] { \(\textsf{Z}[\begin{smallmatrix} \texttt{1} \\ \texttt{1} \end{smallmatrix}]=15\) };

\node[draw=none,node distance=2cm,right of=v1] { \(\textsf{Z}[1]=7\) };

\end{tikzpicture}

\caption{4-ary wavelet tree (left) and a 4-ary wavelet matrix (right) for the text \texttt{accessandselect} over the alphabet \(\{\texttt{a}~\textcolor{gray}{(\texttt{000})_2},\texttt{c}~\textcolor{gray}{(\texttt{001})_2},\texttt{d}~\textcolor{gray}{(\texttt{010})_2},\texttt{e}~\textcolor{gray}{(\texttt{011})_2},\texttt{l}~\textcolor{gray}{(\texttt{100})_2},\texttt{n}~\textcolor{gray}{(\texttt{101})_2},\texttt{s}~\textcolor{gray}{(\texttt{110})_2},\texttt{t}~\textcolor{gray}{(\texttt{111})_2}\}\) (bit representation of the characters is given in gray), i.e., the same input text as in \cref{fig:example_wt_wm}.}
  \label{fig:example_4_ary_wt_wm}
\end{figure*}}

\longversion{\subsection*{Alternative Representations of Sequences}}
\longversion{There exist other compressed self-indices that can answer rank and select queries.}
Recently, a practical block tree implementation has been introduced~\cite{BelazzouguiCGGK2021BlockTrees}.
A block tree is especially useful for highly compressible text, as they require only \(O(z\log(n/z))\) words space, where \(z\) is the number of Lempel-Ziv factors of the text.
\longversion{Unfortunately, even highly tuned implementations are slow to compute~\cite{KopplKM2023LPFBlockTrees}.}
Further dictionary-compressed representations allow for rank and select support in optimal time in compressed space~\cite{Prezza2019RankSelectCompressed} with respect to the size of a string attractor~\cite{Prezza2018StringAttractors} of the text.
For a grammar of size \(g\) and an alphabet of size \(\sigma\), rank and select support requires \(O(\sigma g)\) space \cite{BelazzouguiCPT2015GrammarCompressedRankSelect,PereiraNB2017GrammarCompressedRankSelect}.
Here, queries can be answered in \(O(\log n)\) time.

\longandshortversion{\section{4-Ary Wavelet Trees}}{\section{4-Ary Wavelet Trees and Quad Vectors}}
\label{sec:four_ary_wavelet_trees}
When answering queries using a wavelet tree in practice, the query is translated to \(O(\log\sigma)\) binary rank and select queries.
\longversion{Most of the time to answer a query on the wavelet tree is spent answering these binary rank and select queries.}
\longandshortversion{Additionally, on}{On} each level of the wavelet tree, the binary rank and select queries will result in at least one cache miss, which \longversion{again }is where most of the time for answering a binary rank or select query is used for.
To reduce the number of cache misses, we have to reduce the number of levels.
To this end, we make use of 4-ary wavelet trees.
By doubling the number of children, we (roughly) halve the number of levels.
If \(\ceil{\log\sigma}\) is odd, the 4-ary wavelet tree has \(\ceil{\ceil{\log\sigma}/2}\) levels.

In a 4-ary wavelet tree, we represent the characters on each level using two bits that we store in a quad vector, i.e., a vector over the alphabet \(\{0,1,2,3\}\) with access, rank, and select support\longversion{, see \cref{sec:quad_vectors}}.
If \(\ceil{\log\sigma}\) is odd, characters on the last level are represented using a single bit in a bit vector. 
In the first level, each character is now represented by its two MSBs and all characters share a length-0 bit-prefix.
When visiting the first child of a node that represents characters with bit prefix \(\alpha\), its four children represents characters with bit-prefix \(\alpha\texttt{00}\), \(\alpha\texttt{01}\), \(\alpha\texttt{10}\), and \(\alpha\texttt{11}\) \longversion{(first to fourth child)}.
\longversion{Alternatively, you can say that the \(i\)-th child represents characters that are in the \(i\)-th quarter of the alphabet represented by its parent.See \cref{fig:example_4_ary_wt_wm} for an example.}\longandshortversion{

There also exists a ``4-ary'' wavelet matrix representation of the 4-ary wavelet tree.
Here, we also use two bits to represent the characters at each level.
Again, the first level of the wavelet matrix is the same as the first level of the 4-ary wavelet tree.
Then, to compute the next level \(\ell\) starting with the second, the text is stably sorted using the \((2\ell-1)\)-th and \(2\ell\)-th MSBs as key.
As with the ``normal'' wavelet tree and wavelet matrix, this results in the same intervals, where characters with the same bit-prefix are represented, just in a different order.
Also, queries for the wavelet matrix can be adopted to work with the (4-ary) wavelet matrix in the same way we adopted the queries of the wavelet tree.
To do so, we now have to store the exclusive prefix sum of the histogram of all entries of all levels (as a replacement of \(\textsf{Z}\) in the ``normal'' wavelet matrix).}{Similarly to the binary case, there exist 4-ary wavelet matrices.}

\longversion{
\subsection*{Queries in 4-Ary Wavelet Trees}
Querying a 4-ary wavelet tree works similarly to querying a ``normal'' wavelet tree.
Since there are now four children, more book keeping is necessary to keep track of the interval that is visited during the query.
This does not result in more rank and select queries on the quad vectors (and possibly bit vector on the last level).
Overall, the additional book keeping is less expensive than the cache misses on each level as we can clearly see in our experiments in \cref{sec:experimental_evaluation_wavelet_trees}.
Then again, querying the 4-ary wavelet matrix works similarly to querying the ``normal'' wavelet matrix.
See \cref{ex:4ry_rank_wavelet_matrix} for an example.}

\longversion{\section{Quad Vectors}}
\label{sec:quad_vectors}
At the heart of our 4-ary wavelet trees is a space-efficient and fast rank and select data structure for quad vectors.
Our data structure uses a block-based design and follows the popular memory layout for block-based rank and select data structures for bit vectors~\cite{Kurpicz2022RankSelect,ZhouAK2013PopcountRankSelect} adapted to quad vectors.
In a block-based design, the number of occurrences of different symbols is stored for blocks of different size.
The number is stored either for the whole input up to the block or for the input contained in a bigger block.
For our quad vector, we store the following information for each symbol \(\alpha\in\{\texttt{00},\texttt{01},\texttt{10},\texttt{11}\}\):
\longandshortversion{
  \begin{description}[itemsep=-2mm]
  \item[Super Blocks] cover 4096 symbols and store the number of occurrences before the start of the super block.
  \item[Blocks] cover 512 symbols and store the number of occurrences before the start of the block within the super block.
  \end{description}
For each super block, we only have to store seven blocks, as there are no occurrences of any symbol before the first block within the super block, i.e., all counters are zero.
The counter within each block can be stored in just 12 bits, as the maximum number of occurrences of a single symbol within one super block before the last block is 3584 (\(\ceil{\log 3584}=12\)).
Therefore, the counters of the seven blocks fit into 84 bits and can use 44 bits for the counter of the super block, when using 128 bits for both super block and the pertinent blocks.\longversion{\footnote{In practice, computer words have size 8, 16, 32, 64, 128, and on modern machines even 256 and 512 bits. Aligning the size of (super-)blocks with computer words improves the performance.}}
Additionally, storing super blocks and pertinent blocks interleaved, reduces the number of cache misses and allows for the usage of vectorized instructions~\cite{Kurpicz2022RankSelect}.
\longandshortversion{}{We can reduce the number of cache misses by doubling the required space.
A cache line on nearly all hardware has size 64 bytes.
To make a super block and its pertinent blocks fit into one cache line, the number of symbols per super block and block can be halved.%
This doubles the number of counters we have to store, but we also guaranteed at most two cache misses per rank query and three per select query.
However, by doing so, the space-overhead increases to \(12.5\,\%\).}}{
  \emph{Superblocks} cover 4096 (or 2048) symbols and store the number of occurrences before the start of the super block.
  \emph{Blocks} cover 512 (or 256) symbols and store the number of occurrences before the start of the block within the super block.
  The smaller size (super)blocks result in double the space-overhead but halve the cache misses, as the pertinent information fits into one cache line.
Due to the page limit and this being a minor part of this paper based on~\cite{Kurpicz2022RankSelect,ZhouAK2013PopcountRankSelect}, we do not go into full detail here, as details are not required to understand the remaining part of the paper.
We refer to the full paper~\cite{Ceregini2023FasterWTQueriesFullVersion} for a more detailed description.
}

\longversion{
\begin{lemma}
  A quad vector with rank support has a space-overhead of 6.25\,\%.
  Adding select support introduces an additional overhead of \(4\ceil{\log n}/8192\).
\end{lemma}
\begin{proof}
  We require 128 bits for each super blocks including its blocks for each symbol.
  Resulting in a space-overhead of \(4\cdot 128/8192=6.25\,\%\).
  Since we store the block of every 8192-nd occurrences of a symbol to answer select queries more efficiently, this introduces another \(4\cdot 32/8192=1.5625\,\%\) space-overhead (allowing select support for quad vectors of length up to \(2^{41}\) quads).
\end{proof}

We can reduce the number of cache misses by doubling the required space.
A cache line on nearly all hardware has size 64 bytes.
To make a super block and its pertinent blocks fit into one cache line, the number of symbols per super block and block can be halved.%
This doubles the number of counters we have to store, but we also guaranteed at most two cache misses per rank query and three per select query.
However, by doing so, the space-overhead increases to \(12.5\,\%\).

In theory, we can save space by storing only information for three symbols and using sophisticated data structures to represent the information.
However, we did not implement the following variant, as preliminary experiments showed that the space-saving features heavily impacted the query performance.%

\begin{lemma}
  A quad vector with rank support requires only 2.41\,\% space-overhead.
\end{lemma}
\begin{proof}
We only save the information for three of the four symbols, as we can compute all information for the fourth symbol using the information of the other three symbols.
Removing the information for one symbol saves 25\,\% of space, hence the space-overhead is now only \(4.6875\,\%\).

Using the Elias-Fano encoding~\cite{Elias74,Fano71}, a monotonic increasing sequence of \(k\) integers in a universe of size \(u\) can be stored using only \(k(2+\log(u/k))\) bits while allowing constant time access to all integers.
Since the number of occurrences of symbols withing super blocks are monotonic increasing sequences, we can use Elias-Fano encoding to store them.
To this end, we introduce \emph{mega blocks} that cover \(2^{18}\) symbols.
We store the number of occurrences of each symbol from the beginning of the text to the beginning of each mega block and encode only the information for the remaining three symbols.
We now can store the following information for each super block:
Three 18-bit counters for three symbols storing the number of occurrences from the beginning of the mega block and the Elias-Fano encoded sequences that require at most 141 bits.
Overall, we require 195 bits for \emph{all three} symbols.
Therefore, this variant has a space overhead of 2.41\,\%.
\end{proof}}

\longversion{
\subsection{Answering Queries}
Answering queries using this approach is similar to the bit vector case.
Assume we want to get the rank of the symbol \(\alpha\) at position \(i\).
We simply have to identify the super block (\(i/4096\)) and the block (\((i \mod 4096)/512\)) where the position occurs in.
Adding up these counters, we only have to scan \((i \mod 512)\) positions within the block and add the number of occurrences of \(\alpha\) in the block up to that position to the result.
All this can be done in constant time.
To find the position where the \(j\)-th \(\alpha\) occurs, we first identify the closest smaller sampled position.
Starting from there, we do scan super blocks until we have identified the super block containing the position.
Then, we continue with scanning block until we have identified the block containing the position.
Finally, we scan the quad vector (within the block) until we have found the correct position and return the index.
While this may not be a constant time query, it is very fast in practice, see~\cref{sec:exeperimental_evaluation}.
}

\colorlet{QWT256rank_latency}{my-dark-red}
\colorlet{QWT256rank_prefetch_latency}{my-dark-red}
\colorlet{QWT256pfsrank_latency}{my-deep-purple}
\colorlet{QWT256pfsrank_prefetch_latency}{my-deep-purple}
\colorlet{QWT512rank_latency}{my-indigo}
\colorlet{QWT512rank_prefetch_latency}{my-indigo}
\colorlet{QWT512pfsrank_latency}{my-teal}
\colorlet{QWT512pfsrank_prefetch_latency}{my-teal}
\colorlet{pasta_wmrank_latency}{my-teal}
\colorlet{sdsl_wmrank_latency}{my-blue}
\colorlet{sdsl_fbb_latency}{my-deep-orange}
\colorlet{sucds}{my-green}

\pgfplotsset{
  QWT256/.style={
    color=QWT256rank_latency,
    mark=x,
  },
  QVec256/.style={
    QWT256,
  },
  QWT256pfs/.style={
    color=QWT256pfsrank_latency,
    mark=o,
  },
  QWT512/.style={
    color=QWT512rank_latency,
    mark=square,
  },
  QVec512/.style={
    QWT512,
  },
  QWT512pfs/.style={
    color=QWT512pfsrank_latency,
    mark=triangle,
  },
  pasta_wm/.style={
    color=pasta_wmrank_latency,
    mark=triangle,
  },
  sdsl_wm/.style={
    color=sdsl_wmrank_latency,
    mark=square,
  },
  sdsl_fbb/.style={
    color=sdsl_fbb_latency,
    mark=o,
  },
  fbb_construction/.style={
    sdsl_fbb,
  },
  sucds/.style={
    color=sucds,
    mark=+,
  },
  pasta_bv/.style={
    pasta_wm,
  },
  sdsl_bv/.style={
    sdsl_wm,
  }
}

\section{Faster Rank Queries with Prefetching}\label{sec:prefetching}

Modern CPUs can issue multiple memory requests concurrently, paving the way for proactive prefetching of cache lines predicted to be accessed in the near future.
By issuing the memory requests for the accessed cache line and the anticipated ones simultaneously, prefecthing helps hiding memory latency and reducing the impact of memory access delays on the CPU's execution pipeline.

Prefetching manifests in two forms: \emph{hardware} and \emph{software} prefetching.
Hardware prefetching is implemented within the CPU's microarchitecture and is driven by the hardware itself.
\longversion{Modern CPUs come equipped with dedicated prefetcher units that monitor memory access patterns and automatically issue prefetch requests based on these patterns.
These units analyze the memory addresses being accessed and attempt to predict future memory accesses.
They then fetch the predicted data into the cache in advance.}
\longversion{For example, the \emph{sequential prefetching} predicts that the next memory location to be accessed will be contiguous to the current one.
It fetches additional cache lines in advance to take advantage of spatial locality. This is particularly effective for array traversal where data tends to be stored sequentially. 
The \emph{strided prefetching} instead looks at the delta between the addresses of the memory 
accesses and looks for patterns within it.
If a consistent pattern in the stride is detected, the CPU fetches cache lines based on this pattern, assuming that the algorithm will continue accessing memory addresses with the same stride.

} Software prefetching, instead, is controlled by the programmer or the compiler through explicit instructions.
\longversion{Programmers can insert prefetch instructions (e.g., \texttt{\_mm\_prefetch} intrinsic for x86 CPUs) into their code to indicate which data should be pref etched and when.}
However, \longandshortversion{software prefetching}{it} requires a deep understanding of the algorithm's memory access patterns and the underlying memory hierarchy, because incorrect or excessive prefetching can lead to performance degradation. 

\begin{algorithm2e}[t]
  \Fn{Rank{\(_\alpha(i)\)}}{

    \(r_0 = i,\;b_0=0\)\;
    \For{\(k=1,\dots,\ell+1\)}{
      \(\alpha_k = (\alpha\;>>\;2 * (\ell - 1 - k))\;\&\;3\), \(\texttt{offset} = C_k[\alpha_k]\)\;
      \(b_k = Q[k].\rank_{\alpha_k}(b_{k-1})+ \texttt{offset}\)\label{ex:prefetch_bk}\;
      \(r_k = Q[k].\rank_{\alpha_k}(r_{k-1})+ \texttt{offset}\)\;
    }
    \Return \(r_\ell-b_{\ell}\)\;
  }
  \vspace{.25cm}
  \caption{Rank query for a \emph{4-ry wavelet matrix} with $\ell$ levels. For level \(k\), $Q[k]$ is the quad vector and \(C_k[\alpha_k]\) is the number of character  $<\alpha_k$ on level \(k\).
    \label{ex:4ry_rank_wavelet_matrix}}
\end{algorithm2e}

The goal of this section is to show how to introduce software prefetching in the 
algorithm of the \rank\ query.
For the following discussion, we give the pseudo code for $\rank_\alpha(i)$ query on a 4-ry wavelet matrix in \cref{ex:4ry_rank_wavelet_matrix}. 
A $\rank_\alpha(i)$ query on a wavelet tree has to traverse
each of the $\ell = \ceil{\ceil{\log \sigma}/2}$ levels. 
At each level $k$, we perform two \rank\ queries on the quad vectors of that level for the character $\alpha_k\in [0,3]$ to compute $b_k$ and $r_k$. These two \rank\ queries use the results $b_{k-1}$ and $r_{k-1}$
of the two \rank\ queries computed at the previous level.
\longversion{Below we discuss only how to perform the prefetching for $r_k$ as the prefetching for $b_k$ can be done in a similar way.

}
Every \rank\ query in a quad vector for a given position $i$ needs to access only two cache lines: the one containing counters for the superblock and block of that position, and the one containing \longversion{the data block with }the $i$-th character. 
These two cache lines can be requested in parallel as they only depend on position $i$.
\longversion{Hence, we observe just the latency of at most one cache miss per level. 
The challenge in eliminating this cache miss is that both the cache lines we need to access at a certain level $k$ depend on the position $r_{k-1}$, which is known only when the \rank\ query at level $k-1$ has been computed.}
\longversion{Solving this challenge could be possible with a predictive model capable of anticipating the cache lines required for ranking at position ${r_{k-1}}$ across all levels $k$, 
way before position ${r_{k-1}}$ is computed\longandshortversion{. Such an advanced prediction}{, as it} would enable us to initiate simultaneous requests for all these cache lines.}

\subsection{Predicting Cache Lines in a Quad Vector}
This challenge led us to the definition of the \emph{Rank with Additive Approximation} problem and our predictive model will take the form of a lightweight data structure.

\begin{definition}
Given a quaternary vector $Q[1,n]$ and fixed an additive error \error, 
the goal is to build a data structure to answer additive approximated rank 
queries. Given a position $i$ and a symbol $\alpha \in [0,3]$, 
$\arank_\alpha(i)$ approximates the correct rank query by returning any arbitrary 
value $\tilde{r}$ within $[r, r + \error]$, where $r=\rank_\alpha(i)$.  
\end{definition}

A prediction model that correctly predicts the needed cache lines of a certain level, is actually solving the Rank with Additive Approximation problem on the quad vector of the previous level with $\error$ equal to the cache line size\longversion{, e.g., $512$ bits ($256$ quads)}
\longandshortversion{. Vice versa, if we have a solution for the problem with the same $\error$, we have a way to predict the required cache lines.}{ and vice versa.}

\begin{lemma}\label{lemma:lb}
Any data structure that solves the Rank with Additive Approximation problem on $Q[1,n]$ with additive error $\error$ needs at least $\Omega(n/\error)$ bits of space.
\end{lemma}
\begin{proof}
  Assume by contradiction that there exists a solution for the problem that uses $o(n/\error)$ bits of space for any quad vector of length $n$. 
  Then, we could use this data structure to represent any quad vector with less than $2n$ bits, which is impossible because of an information-theoretical lower bound.

  Given any $Q[1,n]$, we obtain its expanded version of $\hat{Q}[1, 3\error n]$ by replacing each character with a run of $3\error$ of its copies.
  We use the above data structure to index $\hat{Q}$ using $o(n)$ bits of space. 
  Now, we reconstruct $Q$ by querying the data structure for any character at the beginning and the end of each run.
  The correct character in $Q$ can be identified because the results of the two queries differ by at least $2\error$, while the results for the other characters differ by at most $\error$.
\end{proof} 

\begin{lemma}\label{lemma:ub}
There is a data structure with constant query time requiring $\Theta(n/\error)$ bits, i.e., matching the space lower bound, for the Rank with Additive Approximation problem on $Q[1,n]$ with additive error $\error$.
\end{lemma}
\begin{proof}
The idea is to use a bit vector $B_\alpha[1,\ceil{2n/\error}]$, for each of the character  $\alpha \in [0,3]$. We split $Q[1, n]$ into blocks of size $\error/2$. 
The $i$th bit in $B_\alpha$ is set to $1$ if and only if the $i$th block of $Q$ contains the $j$th of $\alpha$, for some $j$ which is a multiple of $\error/2$.

We add the required extra data structure to support $\rank$ queries on the bit vector $B_\alpha$.
A query $\arank_\alpha(i)$ is solved as follows. 
Let $j=\floor{2i/\error}$ be the block in $Q$ that contains our target position $i$.
We compute $k=\rank_1(j-1)$ on the bit vector $B_\alpha$. 
This way, we know that the number of occurrences in $Q$ up to position $i$ is at least
$r \cdot \error/2$. 
Moreover, the exact number of occurrences of $\alpha$ up to the block $j$ is at most $k \cdot \error/2+\error/2-1$. As the $j$th block as size $\error/2$, we conclude that returning $\tilde{r} = k \cdot \error/2$ gives the required estimate.
\end{proof}

\subsection{Predicting Cache Lines in a Wavelet Tree}
Let us consider the rank query $\rank_\alpha(i)$. Consider the rank query $\rank_\alpha(i)$.
For addressing this query through a 4-ary wavelet tree, we divide the character $\alpha$ into its quaternary components $\alpha_1, \alpha_2, \ldots, \alpha_{\ell}$.
Then, at level $k$, we compute $r_k=\rank{\alpha_{k}}(r_k-1)$.
See \cref{ex:4ry_rank_wavelet_matrix}.
As we mentioned above we focus on prefetching for $r_k$s (\cref{ex:prefetch_bk}), as we can deal with $b_{k}$s in a similar way.\longversion{

} The prefetching is possible if can approximate each $r_k$ with $\tilde{r}_{k}$, 
such that $\tilde{r}_{k} \in [r_k, r_k + \error]$ with $\error = 256$.
Indeed, each cache line has size $512$ bits and, thus, spans $256$ positions of the quad vector 
at level $k$. 
The value $\tilde{r}_{k}$ introduces uncertainty only within the span of two 
consecutive cache lines.
Note that prefetching is effective only if we compute the approximated ranks $\tilde{r}_{k}$ for 
all the levels. This way we issue the requests for all the required cache lines in parallel 
before starting to use these cache lines to compute the exact ranks $r_k$.

Unfortunately, solving the Rank with Additive Approximation problem with error $\error$ for the quad vector at each level of the wavelet tree is not enough to guarantee that $\tilde{r}_{k}$ is at most at distance $\error$ from $r_{k}$ (i.e., $\tilde{r}_{k} \in [r_k, r_k + \error]$), for all the levels $k$.
This is because the value $\tilde{r}_{k}$ is computed with an approximated rank at position $\tilde{r}_{k-1}$ because the exact position $r_{k-1}$ is unknown, i.e., we can compute $\arank_{\alpha_k}(\tilde{r}_{k-1})$ and not $\arank_{\alpha_k}(r_{k-1})$.
As the position $\tilde{r}_{k-1}$ is already affected by some error, the errors of our approximations sum up level by level. Thus, at level $k$ the error could be up to $(k-1)\error$. 

We can solve this issue by correcting the approximations at each level. 
This approach is inspired by a solution for the substring occurrence estimation on texts with compressed indexes \cite{ALGO15}.
The main idea is to refine the estimates at each level $k$ with a correction term $\Delta$.
To compute $\Delta$ we need to store a set of \emph{discriminant} positions $D_{k,\alpha}$ for each character $\alpha \in [0,3]$ at level $k$. 

In the solution of \cref{lemma:ub} we store a bit vector $B_\alpha$ for each character $\alpha \in [0,3]$.
A bit was set to one for each position corresponding to an occurrence of $\alpha$ which is a multiple of $\error$. 
The set $D_{k,\alpha}$ consists of the position in the quad vector corresponding to those occurrences. 
\longandshortversion{We note the positions in these sets can be stored within $\Theta(\log \error)$ bits per position in several ways.
The most suitable one for our purposes is to associate each position with its corresponding bit set to one in $B_\alpha$ and store its offset within the corresponding block.}{The positions in these sets can be stored within $\Theta(\log \error)$ bits per position, e.g., by associating each position with its corresponding bit set to one in $B_\alpha$ and store its offset within the corresponding block.}

At query time, given $r_{k-1}$ and the character $\alpha_k$, we want
to compute the discriminant position $d_{k-1}$ which is the successor of $r_{k-1}$
in the set $D_{k,{\alpha_k}}$. This discriminant position can be computed in constant 
time with a \rank\ and a \select\ query on the bit vector of $\alpha$.
Once we computed $d_{k-1}$, the correction term $\Delta$ is $\min(d_{k-1} - \tilde{r}_{k-1}, \error-1)$ and the approximated rank is computed as $\tilde{r}_k = \rank_{\alpha_k}(d_{k-1}) - \Delta$.
This correction is enough to guarantee that our approximations always remain at a distance at most $\error$ from the correct ones over all the levels $k$ of the wavelet tree.

\begin{lemma}
 At any level $k$, we have $\tilde{r}_{k} \in [r_k, r_k+\error)$.
\end{lemma}

\begin{proof}
  The proof is by induction on $k$. 
  For the first level $k=1$, as at the beginning $\tilde{r_0} = r_0$, we have $r_1 \in [r_1, r_1 +\error)$ by \cref{lemma:ub}.
  For general $k$, we assume that $\tilde{r}_{k-1} \in [r_{k-1}, r_{k-1}+\error)$, and we prove $\tilde{r}_{k} \in [r_{k}, r_{k}+\error)$. 
  We want to prove that $\tilde{r}_{k} \leq r_k$ and $r_k - \tilde{r}_{k} \leq \error$.
  There are two cases based on the relationship between $d_{k-1}$ and $r_{k-1}$.
  By definition we know that $\tilde{r}_{k-1} \leq d_{k-1}$ and by inductive hyphotesis $\tilde{r}_{k-1} \leq r_{k-1}$.
  
  The first case is $r_{k-1} \leq d_{k-1}$.
  Thus, we have $\tilde{r}_{k-1} \leq r_{k-1} \leq d_{k-1}$.
  Let $z$ be the number of occurrences of the ranked character $\alpha_k$ in the interval $[r_{k-1}, d_{k-1}]$.
  \longandshortversion{
  $$
  \begin{array}{lcl}
    r_k - \tilde{r}_{k} & = & \rank_{\alpha_k}(r_{k-1}) - \arank_{\alpha_k}(\tilde{r}_{k-1})\\
                              & = & \rank_{\alpha_k}(d_{k-1}) - z - (\rank_{\alpha_k}(d_{k-1}) - \Delta) \\
                              & = & \Delta - z \leq \error \\
  \end{array}
  $$}{
  Now, we have \(r_k - \tilde{r}_{k}=\rank_{\alpha_k}(r_{k-1}) - \arank_{\alpha_k}(\tilde{r}_{k-1}) = \rank_{\alpha_k}(d_{k-1}) - z - (\rank_{\alpha_k}(d_{k-1}) - \Delta) = \Delta - z \leq \error\).}
  The last inequality follows by $[r_{k-1}, d_{k-1}]$ being contained in $[\tilde{r}_{k-1},d_{k-1}]$, bounding $z$ by the minimum of the length of $[\tilde{r}_{k-1},d_{k-1}]$ and $\error-1$.
  If the interval is larger than $\error-1$, there cannot be more than $\error-1$ of $\alpha_k$ since we sampled a discriminant position every $\error$ occurrences of $\alpha_k$.
  It also follows that $z \leq \Delta$ and, thus, $\tilde{r}_{k} \leq r_k$.

  The second case is $d_{k-1} < r_{k-1}$.
  Thus, $\tilde{r}_{k-1} \leq d_{k-1} \leq r_{k-1}$.
  Let $z$ be the number of occurrences of $\alpha_k$ in the interval $[r_{k-1}, d_{k-1}]$.
  \longandshortversion{
  $$
  \begin{array}{lcl}
    r_k - \tilde{r}_{k} & = & \rank_{\alpha_k}(r_{k-1}) - \arank_{\alpha_k}(\tilde{r}_{k-1})\\
                              & = & \rank_{\alpha_k}(d_{k-1}) + z - (\rank_{\alpha_k}(d_{k-1}) - \Delta) \\
                              & = & z + \Delta \leq r_{k-1} - \tilde{r}_{k-1} \leq \error  \\
  \end{array} 
  $$}{
Now, we have \(r_k - \tilde{r}_{k} = \rank_{\alpha_k}(r_{k-1}) - \arank_{\alpha_k}(\tilde{r}_{k-1})
                              = \rank_{\alpha_k}(d_{k-1}) + z - (\rank_{\alpha_k}(d_{k-1}) - \Delta)
                              = z + \Delta \leq r_{k-1} - \tilde{r}_{k-1} \leq \error\).
  }
  The first inequality follows by observing that $\Delta$ is at most the distance between 
  $\tilde{r}_{k-1}$ and $d_{k-1}$ and $z$ is at most the distance between $d_{k-1}$ and
  $r_{k-1}$. So, their sum is at most $r_{k-1} - \tilde{r}_{k-1}$. The last inequality is by inductive hypothesis.
\end{proof}

The space required by this predicting data structure is $\Theta((n/\error)\log \error)$ for each level of the wavelet tree.
So, the overall space usage is $\Theta((n\log \sigma/\error)\log \error)$ bits.
As we mentioned above, prefetching with the above data structure can be done by setting $\error=256$.
However, we are left with an issue.
If the indexed sequence is too large, the predicting data structure itself does not fit in the cache and, thus, to avoid cache misses in the wavelet tree we would pay cache misses in the predicting data structure. 
This issue could be solved by introducing a hierarchy of predictors in which a predictor at a specific level takes on the responsibility of prefetching the necessary cache lines for the subsequent-level predictor.
Each predictor allows an error that is roughly $\error$ times less than the one at the next level, until the predictor at the head of the hierarchy fits in cache.
Unfortunately, a larger hierarchy becomes impractical quite soon for two reasons. 
First, to fully exploit prefetching we would have to request all the predicted cache lines in parallel, and current CPUs can issue only 5--10 memory requests in parallel.
Second, each level of the hierarchy introduces a cost of $\Theta(\log \sigma)$ to the query. 

\longandshortversion{\subsection{Practical Implementations}}{\paragraph{Practical Implementations.}}
In our implementation, we relaxed the previous solution in several respects.
First, we do not use the correcting term $\Delta$ and the discriminant positions. 
This is because in our tests we used sequences with an alphabet size $\sigma$ up to $256$, which requires a wavelet tree of at most $4$ levels.
Thus, the error growth here is very limited and it can be afforded by prefetching more cache lines.
Second, we limit the hierarchy to just two levels of predictors. 
The first one implements the solution of \cref{lemma:ub} with error $\error=2048$. 
For the second level, we observe that super block and block counters can be used as a variant of the solution of \cref{lemma:ub} with error $\error=256$.
\longversion{Even if is more space inefficient, it is preferred because the space is already needed by the wavelet tree implementation.}
This way, we can use the first level to predict the super block containing $r_k$ for each level and prefetch the cache lines containing the counters of those super blocks and their blocks. 
Then, we use these counters to refine the predictions to prefetch the cache lines with the correct blocks of data in the quad vectors.
\longversion{

  The first level of the hierarchy has to store a bit vector of size $n'=n/2048$ bits for each 
character $\alpha\in [0,3]$ in each level of the wavelet tree. We enhance this bit vector 
with \rank\ support with a space overhead of $n'/4$ bits \cite{Vigna2008BroadwordRankSelect}.
So, the overhead of the predicting data structure is $\ell(4n/2048+n/2048) = 5\ell n/2048$ bits, 
where $\ell$ is the number of levels in the 4-ary wavelet tree.
The predictor at the second level of the hierarchy does not introduce any extra space overhead. 
Observe that the first predictor fits in a $32$ MB L3 cache for sequences up to $\approx 100$ Billions of characters over an alphabet of size $256$.}
\longandshortversion{

We conclude by observing that cache lines needed by \access\ and \select\ queries cannot be predicted with proposed solutions.
The reason is that each query on a certain level has a double dependency on the result of the previous one. 
Indeed, both position and symbol are known only when the previous query is solved.}{
 Cache lines needed by \access\ and \select\ queries cannot be predicted with our solution, as for each level there is a double dependency (position and symbol) on the result of the previous level.
}

\section{Experimental Evaluation}
\label{sec:exeperimental_evaluation}
\longversion{
We first discuss our experimental setup.
Then, in \cref{sec:experimental_evaluation_wavelet_trees}, we show the benefit of using 4-ary wavelet trees and approximate rank queries.
Finally, in \cref{subsec:expquad}, we compare quad and bit vectors.

\subsection*{Experimental Setup}}
For our experiments, we used a machine equipped with two AMD EPYC 7713 \longversion{(64 cores, 2 threads per core, 2\,GHz base clock, and cache sizes: 64\,KB L1I and 64\,KB L1D per core, 512\,KB L2I+D per core, and 256\,MB L3I+D, with 32\,MB per 8 cores called \emph{core complex})} and 2\,TB DDR4 RAM running Ubuntu 20.04.3 LTS kernel version 5.4.0-155.
All experiments were performed using a single thread, with hyperthreading and turbo boost disabled.
C++ code of competitors was compiled with GCC 11.1.0 with flags \texttt{03} and \texttt{march=native} and Rust code was compiled using \texttt{cargo build --release}.
Our Rust implementation is available at \url{https://github.com/rossanoventurini/qwt}.
We ran each experiment ten times (10M queries for each run) and report the average running time.
\longversion{We generate all queries in advance as follows.
Let $S[0,n)$ be the indexed sequence.
Each \access\ query asks to access the symbol at a random position in the sequence.
For \rank\ queries, we generate a random position $i \in [0, n)$ and use $\langle i, S[i] \rangle$ as a \rank\ query.
For \select\ queries, we select a symbol $c$ at random following their distribution in the sequence, i.e., more frequent symbols have a higher probability of being selected.
Then, we generate a random value $r \in [1, occ(c)]$, where $occ(c)$ is the number of occurrences of $c$, and use $\langle r, c \rangle$ as a \select\ query.}

\longversion{
\begin{figure*}
  \centering
   \makebox[\textwidth]{
  \input{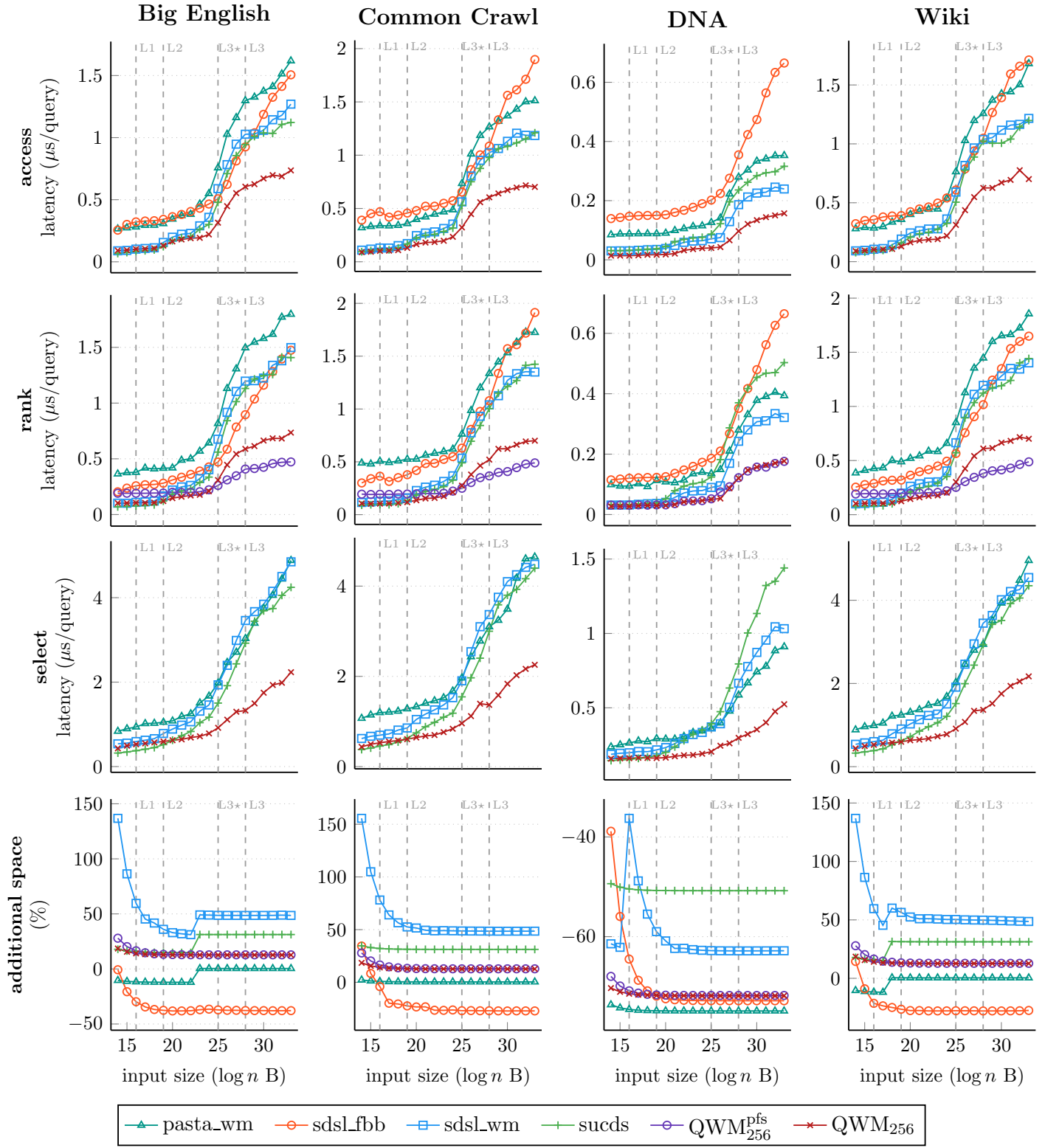}}
  \ref{leg:query_latency}
\caption{Comparison of our new wavelet trees with competitors.
    The first three rows give the access, rank, and select query latency of the implementations, and the last row shows the space overhead of the wavelet tree compared to the input (storing one character per byte).
    The vertical dashed gray lines indicate the L1, L2, L3\(\star\), and L3 cache sizes of the CPU used for these experiments (L3\(\star\) indicates the L3 cache size per core complex).}
  \label{fig:big_wm_latency_comparison}
\end{figure*}}

\begin{table}
  \caption{\emph{Latency} of access, rank, and select queries (row 1--3) given in \(\mu s\) and the \emph{space} (row 4) is given in GiB.
    The small number in parentheses is the \emph{speedup of \wtqvectorpfs{256}} over the method represented by the column.
  All results for 8\,GiB input files.}
  \label{tab:summary_latency}
  \centering

\begin{tabular}{llr@{\hspace{8pt}}rr@{\hspace{8pt}}rr@{\hspace{8pt}}rr@{\hspace{8pt}}rr@{\hspace{8pt}}r}
  \toprule
  & input & \multicolumn{2}{c}{\wtsdsl} & \multicolumn{2}{c}{sdsl\_fbb} & \multicolumn{2}{c}{\wtpasta}  & \multicolumn{2}{c}{sucds} & \multicolumn{2}{c}{\wtqvectorpfs{256}}\\
  \midrule
  \multirow{4}{*}{\rotatebox[origin=c]{90}{access\vphantom{lk}}}
   & English & 1270 & {\footnotesize(1.7\(\times\))} & 1506 & {\footnotesize(2.1\(\times\))} & 1618 & {\footnotesize(2.2\(\times\))} & 1122 & {\footnotesize(1.5\(\times\))} & 731 & {\footnotesize(1.0\(\times\))} \\
   &      CC & 1185 & {\footnotesize(1.7\(\times\))} & 1897 & {\footnotesize(2.7\(\times\))} & 1511 & {\footnotesize(2.2\(\times\))} & 1210 & {\footnotesize(1.7\(\times\))} & 700 & {\footnotesize(1.0\(\times\))} \\
   &     DNA &  239 & {\footnotesize(1.5\(\times\))} &  665 & {\footnotesize(4.2\(\times\))} &  353 & {\footnotesize(2.2\(\times\))} &  316 & {\footnotesize(2.0\(\times\))} & 157 & {\footnotesize(1.0\(\times\))} \\
   &    Wiki & 1216 & {\footnotesize(1.7\(\times\))} & 1712 & {\footnotesize(2.4\(\times\))} & 1681 & {\footnotesize(2.4\(\times\))} & 1198 & {\footnotesize(1.7\(\times\))} & 712 & {\footnotesize(1.0\(\times\))} \\
  \midrule
  \multirow{4}{*}{\rotatebox[origin=c]{90}{rank}}
   & English & 1498 & {\footnotesize(3.2\(\times\))} & 1474 & {\footnotesize(3.1\(\times\))} & 1797 & {\footnotesize(3.8\(\times\))} & 1408 & {\footnotesize(3.0\(\times\))} & 472 & {\footnotesize(1.0\(\times\))} \\
   &      CC & 1350 & {\footnotesize(2.8\(\times\))} & 1913 & {\footnotesize(3.9\(\times\))} & 1725 & {\footnotesize(3.5\(\times\))} & 1424 & {\footnotesize(2.9\(\times\))} & 490 & {\footnotesize(1.0\(\times\))} \\
   &     DNA &  321 & {\footnotesize(1.8\(\times\))} &  665 & {\footnotesize(3.8\(\times\))} &  394 & {\footnotesize(2.2\(\times\))} &  503 & {\footnotesize(2.9\(\times\))} & 176 & {\footnotesize(1.0\(\times\))} \\
   &    Wiki & 1402 & {\footnotesize(2.9\(\times\))} & 1649 & {\footnotesize(3.4\(\times\))} & 1855 & {\footnotesize(3.8\(\times\))} & 1442 & {\footnotesize(3.0\(\times\))} & 488 & {\footnotesize(1.0\(\times\))} \\
  \midrule
  \multirow{4}{*}{\rotatebox[origin=c]{90}{select}}
   & English & 4849 & {\footnotesize(2.2\(\times\))} & \multicolumn{1}{c}{---} & \multicolumn{1}{c}{---} & 4882 & {\footnotesize(2.2\(\times\))} & 4245 & {\footnotesize(1.9\(\times\))} & 2229 & {\footnotesize(1.0\(\times\))} \\
   &      CC & 4483 & {\footnotesize(2.0\(\times\))} & \multicolumn{1}{c}{---} & \multicolumn{1}{c}{---} & 4646 & {\footnotesize(2.1\(\times\))} & 4396 & {\footnotesize(1.9\(\times\))} & 2260 & {\footnotesize(1.0\(\times\))} \\
   &     DNA & 1032 & {\footnotesize(2.0\(\times\))} & \multicolumn{1}{c}{---} & \multicolumn{1}{c}{---} &  910 & {\footnotesize(1.7\(\times\))} & 1440 & {\footnotesize(2.8\(\times\))} &  521 & {\footnotesize(1.0\(\times\))} \\
   &    Wiki & 4546 & {\footnotesize(2.1\(\times\))} & \multicolumn{1}{c}{---} & \multicolumn{1}{c}{---} & 4956 & {\footnotesize(2.3\(\times\))} & 4349 & {\footnotesize(2.0\(\times\))} & 2185 & {\footnotesize(1.0\(\times\))} \\
  \midrule
  \midrule
  \multirow{4}{*}{\rotatebox[origin=c]{90}{space}}
  
   & English & \multicolumn{2}{r}{11.9} & \multicolumn{2}{r}{5.0} & \multicolumn{2}{r}{8.0} & \multicolumn{2}{r}{10.5} & \multicolumn{2}{r}{9.0} \\
   &      CC & \multicolumn{2}{r}{11.9} & \multicolumn{2}{r}{5.8} & \multicolumn{2}{r}{8.0} & \multicolumn{2}{r}{10.5} & \multicolumn{2}{r}{9.0} \\
   &     DNA &  \multicolumn{2}{r}{3.0} & \multicolumn{2}{r}{2.2} & \multicolumn{2}{r}{2.0} &  \multicolumn{2}{r}{3.9} & \multicolumn{2}{r}{2.3} \\
   &    Wiki & \multicolumn{2}{r}{11.9} & \multicolumn{2}{r}{5.8} & \multicolumn{2}{r}{8.0} & \multicolumn{2}{r}{10.5} & \multicolumn{2}{r}{9.0} \\

\bottomrule
\end{tabular}

\end{table}

\longversion{
\subsection{Evaluation of 4-Ary Wavelet Trees}}
\label{subsec:expwt}
\label{sec:experimental_evaluation_wavelet_trees}
\longversion{We compare our 4-ary wavelet trees with other wavelet trees.}
Note that we compare wavelet \emph{matrices} if available, as those are faster in practice than wavelet trees\longversion{, however, all implementations mentioned below contain also support for both wavelet trees}.
\longversion{To the best of our knowledge, there exists no other \(k\)-ary wavelet tree implementation for \(k>2\) and no other (uncompressed) wavelet tree implementation with access, rank, and select support.}
In the following, \emph{\wtsdsl} denotes wavelet matrices built on bit vectors of the SDSL library (\texttt{wm\_int})~\cite{GogBMP2014SDSL}.
We also included the fastest compressed wavelet tree implementation in the SDSL---\emph{sdsl\_fbb}~\cite{GogKKPP2019FixedBlockBoosting}.
A wavelet matrix implementation built on bit vectors of the PASTA-toolbox library, using the most space-efficient rank and select data structures~\cite{Kurpicz2022RankSelect}, is denoted by \emph{\wtpasta}.
Additionally, \emph{sucds} is the wavelet matrix implementation in the sucds library\footnote{\url{https://github.com/kampersanda/sucds}, last accessed 2023-11-08.}.
\emph{\wtqvector{256}} and \emph{\wtqvector{512}} are our implementations of wavelet matrices built on quad vectors with blocks of size 256 and 512 symbols per block, cf. \cref{sec:four_ary_wavelet_trees}.
\emph{\wtqvectorpfs{}} denotes the usage of our predictive model (see \cref{sec:prefetching}).
We wanted to include a wavelet matrix based on learned compressed rank and select data structures~\cite{BoffaFV2022LearnedRankSelect}, however\longversion{, here query times are significantly greater\longversion making the plots harder to read.
Furthermore}, the experiments for inputs \(>1\)\,GiB did not finish in reasonable time.
\longversion{Therefore, we excluded this implementation from the results.}

As inputs, we use text prefixes between 16\,KiB and 8\,GiB in size, generated from the following datasets.
\emph{\english} is the concatenation of all 35\,750 English text files from the Gutenberg Project\longandshortversion{.
We removed the headers related to the project, leaving just the real text.}{ without project related headers.}
\emph{\dna} are FASTQ files from the 1000 Genomes Project, where we considered only the raw sequence and kept only the character \texttt{A}, \texttt{C}, \texttt{G}, and \texttt{T}\longversion{, to obtain a very small alphabet}.
\emph{\cc} is a concatenation of the WET files of Common Crawl corpus,\longandshortversion{ i.e., a web crawl without HTML tags.
Here, we removed all meta information added by the corpus.}{ without project related headers.}
\emph{\wiki} is a concatenation of XML data of the English Wikipedia from June 2023.
\longversion{Note that we did not use the famous Pizza\&Chili corpus, as we needed inputs larger than 2\,GiB.}

\longversion{
\begin{figure*}
  \centering
 \makebox[\textwidth]{
  \input{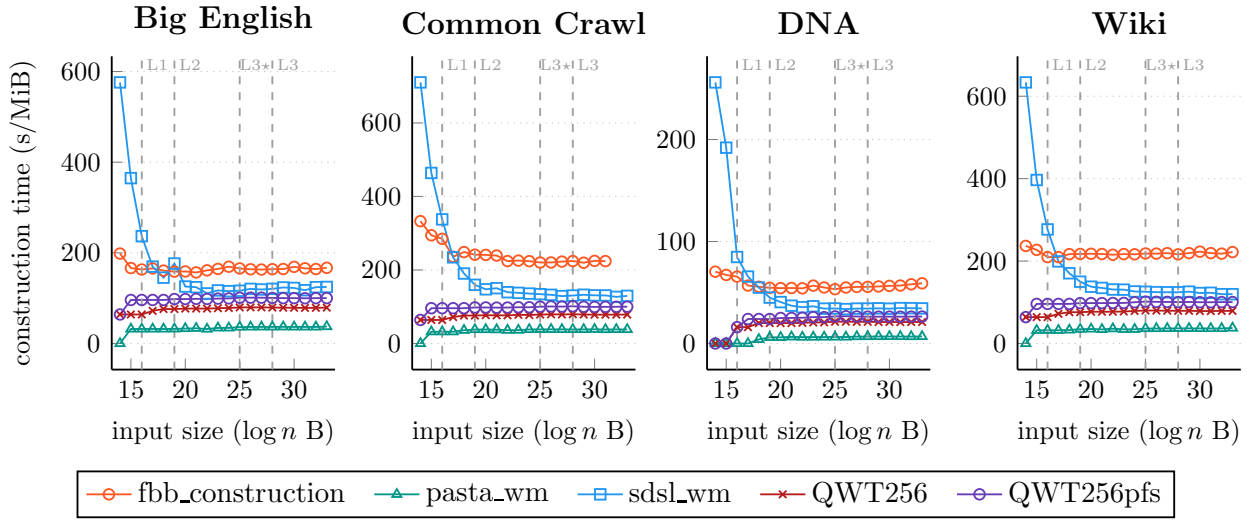}}
\ref{leg:qwm_space}
  \caption{Wavelet tree construction times normalized by input size.
    Vertical dashed gray lines indicate the L1, L2, L3\(\star\), and L3 cache sizes of the CPU used for these experiments (L3\(\star\) indicates the L3 cache size per core complex).\label{fig:wm_construction_time}}
\end{figure*}}

\longandshortversion{
\subsubsection*{Access, Rank, and Select Queries}
The first three rows of \cref{fig:big_wm_latency_comparison} show the query result of our experimental evaluation.
In the main part, we only list results for our wavelet tree with block size 256 as they are never slower than those with block size 512, see \cref{fig:big_wm_self_comparison}.
Overall, sucds almost always at most as fast as \wtsdsl, therefore, to save space, we only discuss \wtsdsl\ in the following.
For \emph{access} and \emph{select} queries, the behavior of all algorithms is very similar.
Here, \wtpasta\ is always the slowest.
For inputs for which the wavelet tree fits into the L2 cache, \wtsdsl\ and \wtqvector{256} have a similar query time.
This is to be expected, as there are still not many cache misses during querying.
There is a steep increase in query time as soon as the wavelet tree does not fit into the L3 cache of the core complex.
However, \wtqvector{256}'s query time does not increase as fast, resulting in a speedup of up to 1.73 (access) and 2.17 (select) compared to \wtsdsl.

Finally, for \emph{rank} queries the behavior of the three previously discussed algorithms is similar.
But here, we also have our wavelet tree \wtqvectorpfs{256} using our prediction mode.
We can actually see it being effective as soon as the wavelet trees doe not fit into the L3 cache of the core complex, i.e., as soon as we expect more and more cache misses during querying.
Here, \wtqvectorpfs{256} becomes faster than \wtqvector{256}, resulting in a speedup of up to 1.55 compared to our wavelet tree without prediction model and up to 3.17 compared to \wtsdsl.
A summary of the results for the largest inputs can be found in \cref{tab:summary_latency}.

In \cref{fig:big_wm_self_comparison}, we show a comparison of our different wavelet tree configurations.
Here, we can see that for access and select queries the query time is very similar for \wtqvector{256} and \wtqvector{512}.
However, for rank queries, \wtqvectorpfs{256} is always faster than \wtqvectorpfs{512}.
Therefore, we only include \wtqvector{256} and \wtqvectorpfs{256} in the main part of this paper.
Note that the space requirements are as expected for all wavelet tree variants.

In \cref{fig:throughput_comparison}, we show a comparison of the throughput of all tested wavelet trees.
Here, we can see that our approach has a similar benefit.
We currently cannot explain the outliers in our measurements for small inputs, however, they do not occur on other hardware and are not easy to reproduce even on our experimental setup.
We are sure that no other heavy process was running at the time of the experiment.
}{
  \paragraph{Experimental Results.}
  Due to space constraints, we mainly consider the latency of \access , \rank , and \select\ by forcing the input of each query to depend on the output of the previous one.
  This is consistent with the in real settings, e.g., the backwards search.
  For a very thorough evaluation, we refer to the full paper~\cite{Ceregini2023FasterWTQueriesFullVersion}, where we also give the throughput and show results for different input sizes.
  Additionally, we only consider \wtqvectorpfs{256} here, as this version is the overall fastest.
  For a comparison of different block sizes (with and without predictive model, please see the full paper).

  We report a summary of our experimental results for inputs of size 8\,GiB in \cref{tab:summary_latency}.
  There, we can see that our new wavelet tree is always the fastest.
  For \access\ and \select\ queries, we achieve a speedup of 1.5--2.2 compared to \wtsdsl, the second fastest wavelet tree.
  When using our predictive model for \rank\ queries, wie can improve this speedup up to 3.2.
  For small alphabets, e.g., \dna, the predictive model provides no advantage, as there is only one level in our 4-ary wavelet tree.
  The reported speedups are in line with other implementations, e.g., sucds, which provides slightly slower \rank\ queries than \wtsdsl.
  
  The space requirements of all wavelet trees are also unsurprising.
  The compressed wavelet tree sdsl\_fbb requires the least space, the space efficient implementation \wtpasta\ requires just a little bit more than the input size, and our new solution is also very space efficient.
  Both, \wtsdsl\ and sucds require slightly more space due to the underlying rank and select data structures.

  Overall, our new wavelet tree provides impressive speedups compared to all other available wavelet tree implementations.
  It is also very space-efficient, i.e., only compressed wavelet trees require significantly less space.
  In the future, we want to integrate our predictive model in compressed wavelet trees.
}
\longversion{
\begin{figure*}
  \centering
   \makebox[\textwidth]{
     \input{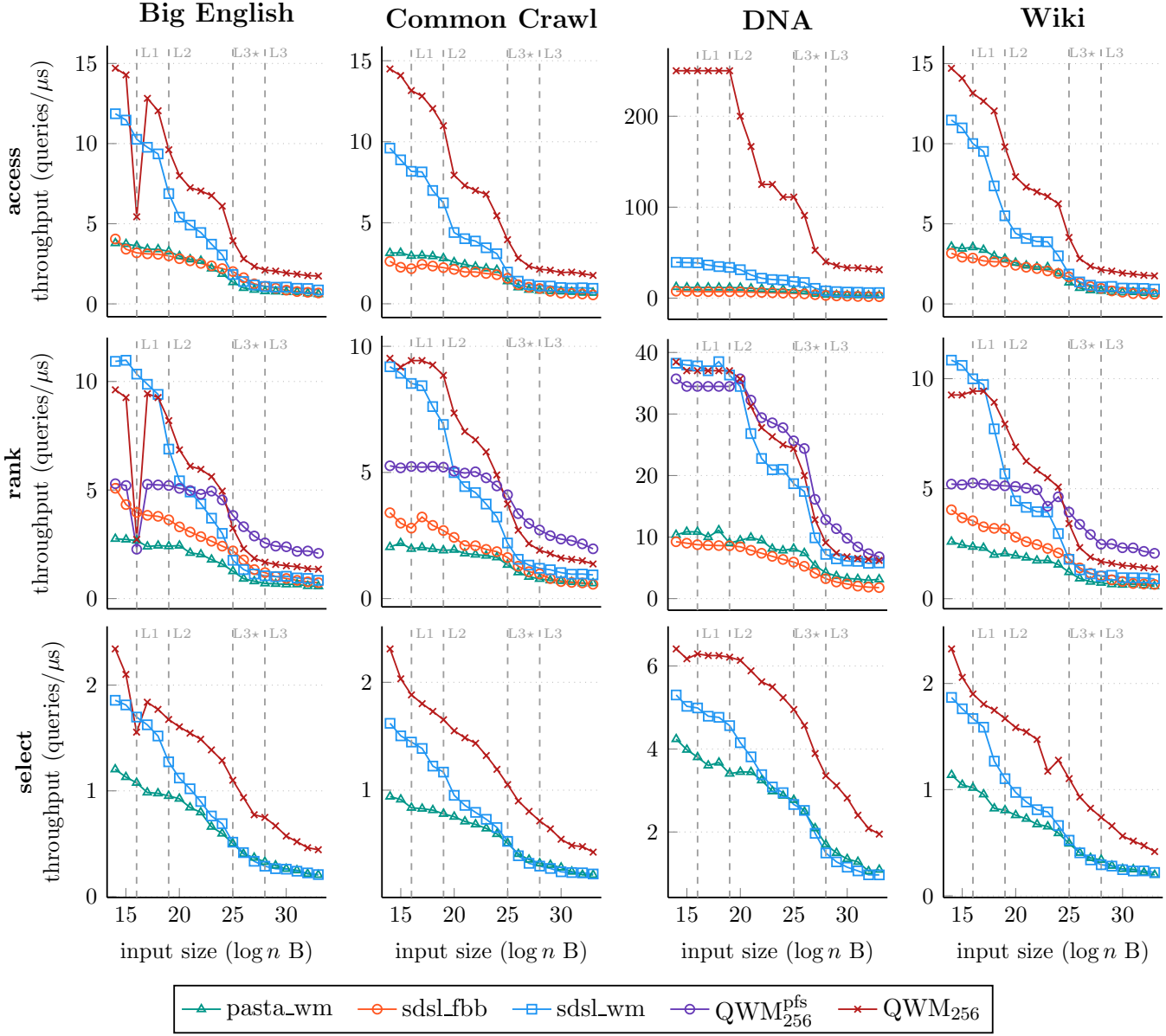}}
     \ref{leg:query_througput}
\caption{Comparison of the \emph{throughput} our new wavelet trees with competitors.
  The vertical dashed gray lines indicate the L1, L2, L3\(\star\), and L3 cache sizes of the CPU used for these experiments (L3\(\star\) indicates the L3 cache size per core complex).
  The additional space is the same as reported in our latency experiments in \cref{fig:big_wm_latency_comparison}.\label{fig:throughput_comparison}}
\end{figure*}
}

\longversion{
\subsubsection*{Space-Overhead and Construction}
The last row in \cref{fig:big_wm_latency_comparison} shows the additional space required by the wavelet tree compared to the input.
Additional space \(<0\,\%\) is due to the compression of the alphabet, where the wavelet tree requires less than 8 levels and the input requires one byte per character.
Similarly, stark increases are due to a new level being requires, which is not the case for our wavelet trees, as we do not use a bit vector for the last level, even if it would suffice, initially wasting some space in the progress.
The spikes in the beginning are due to the general overhead of the data structures.
Overall, the space-overhead is mostly based on the used rank and select data structures.
Since our quad vectors support rank and select using less memory than the bit vectors in the SDSL, our wavelet trees are also more memory efficient---around 75\,\% less space-overhead.
Without surprise, \wtpasta\ is the most space-efficient wavelet tree, however, this comes with slower (wavelet tree) queries.

The construction time of the wavelet trees is shown in \cref{fig:wm_construction_time}.
Here, \wtpasta\ is the fastest to construct, our \wtqvector{256}\ and \wtqvectorpfs{256}\ require similar time but the additional time required for the prediction model is visible, and \wtsdsl\ is the slowest to construct.
Note that we are not focusing on the construction and highly tuned construction algorithms for wavelet trees exist~\cite{DinklageFKT2023SIMDWT}.
}

\longversion{
\subsection{Evaluation of Quad Vectors}
\label{subsec:expquad} 
We now compare our quad vectors with the bit vectors used in the wavelet trees in \cref{sec:experimental_evaluation_wavelet_trees} to show that the speedup is actually due to the improvements presented in this paper and not only due to better rank and select support for quad vectors.
\bvsdsl\ is the implementation of bit vectors of the SDSL library \cite{GogBMP2014SDSL}.
For rank queries, the bit vector is enhanced with {\sf rank\_support\_v}, and it uses {\sf select\_support\_mcl} to support \select\ queries, i.e., Clark and Munro's approach to compute \select\ queries~\cite{ClarkM1996Select}.
\bvpasta\ is the implementation of bit vectors of the PASTA library \cite{Kurpicz2022RankSelect}.%
The bit vector is enhanced using the {\sf FlatRankSelect} data structure to support \rank\ and \select\ queries.
\qvector{256} \textnormal{and} \qvector{512} are our implementations of quad vectors with blocks of size 256 and 512 symbols (see \cref{sec:quad_vectors}).

In these experiments, we generate random bit/quad sequences containing from 16\,K to 8\,G symbols, i.e., the number of symbols contained in each level of the wavelet trees.
Each bit has a 50\,\% chance of being set and each quad has a 25\,\% percent chance of appearing.

}

\longversion{
\begin{figure*}[t]
  \centering
   \makebox[\textwidth]{
  \begin{tabular}{rrrr}
  \begin{tikzpicture}[trim axis right]
    \begin{axis}[
      plotLatencySmall,
      title={\textbf{access}},
      ylabel={\begin{tabular}{c}
                latency (ns/query)
              \end{tabular}},
            legend to name={leg:qvec_latency},
            legend style={font=\small},
            legend columns=4,
            xlabel={input size (\(\log n\) symbols)},
            ]
\addplot[pasta_bv] coordinates { (14.0,9.52117) (15.0,9.55379) (16.0,9.55736) (17.0,9.56679) (18.0,10.2893) (19.0,11.8868) (20.0,12.7143) (21.0,13.31) (22.0,18.514) (23.0,23.5627) (24.0,29.1869) (25.0,32.8885) (26.0,35.3779) (27.0,41.4306) (28.0,62.2165) (29.0,91.8615) (30.0,111.129) (31.0,125.027) (32.0,139.893) (33.0,147.965) };
\addlegendentry{pasta\_bv};
\addplot[sdsl_bv] coordinates { (14.0,12.1518) (15.0,11.5547) (16.0,11.5432) (17.0,11.5961) (18.0,12.2582) (19.0,13.8922) (20.0,14.6933) (21.0,15.19) (22.0,18.7986) (23.0,25.5438) (24.0,32.6397) (25.0,34.9444) (26.0,36.2459) (27.0,45.3676) (28.0,63.6737) (29.0,95.1103) (30.0,113.288) (31.0,147.688) (32.0,142.547) (33.0,146.162) };
\addlegendentry{sdsl\_bv};
            
\addplot[QVec256] coordinates { (14.0,9) (15.0,9) (16.0,9) (17.0,9) (18.0,11) (19.0,12) (20.0,12) (21.0,16) (22.0,24) (23.0,30) (24.0,32) (25.0,33) (26.0,41) (27.0,59) (28.0,87) (29.0,111) (30.0,126) (31.0,137) (32.0,145) (33.0,152) };
\addlegendentry{QVec256};
\addplot[QVec512] coordinates { (14.0,9) (15.0,9) (16.0,9) (17.0,9) (18.0,11) (19.0,12) (20.0,12) (21.0,16) (22.0,26) (23.0,30) (24.0,32) (25.0,33) (26.0,43) (27.0,61) (28.0,89) (29.0,112) (30.0,138) (31.0,137) (32.0,137) (33.0,153) };
\addlegendentry{QVec512};

 \end{axis}
\end{tikzpicture}&
   \begin{tikzpicture}[trim axis right]
    \begin{axis}[
      plotLatencySmall,
      title={\textbf{rank}},
      xlabel={input size (\(\log n\) symbols)},
    ]

\addplot[pasta_bv] coordinates { (14.0,41.129) (15.0,41.1744) (16.0,41.1561) (17.0,40.8653) (18.0,41.1103) (19.0,41.147) (20.0,41.4692) (21.0,41.6291) (22.0,46.0854) (23.0,51.9721) (24.0,58.4905) (25.0,62.5566) (26.0,65.0099) (27.0,72.3063) (28.0,99.4893) (29.0,137.127) (30.0,161.168) (31.0,175.224) (32.0,193.465) (33.0,200.153) };
\addlegendentry{pasta\_bv};
\addplot[sdsl_bv] coordinates { (14.0,17.5005) (15.0,16.915) (16.0,17.5796) (17.0,17.1031) (18.0,18.6856) (19.0,19.8607) (20.0,20.3698) (21.0,21.1818) (22.0,31.2244) (23.0,38.1892) (24.0,41.239) (25.0,42.234) (26.0,44.2543) (27.0,67.734) (28.0,102.147) (29.0,129.037) (30.0,143.204) (31.0,177.841) (32.0,162.796) (33.0,166.417) };
\addlegendentry{sdsl\_bv};
    
\addplot[QVec256] coordinates { (14.0,13) (15.0,13) (16.0,13) (17.0,14) (18.0,14) (19.0,15) (20.0,16) (21.0,21) (22.0,29) (23.0,34) (24.0,36) (25.0,38) (26.0,48) (27.0,73) (28.0,104) (29.0,130) (30.0,143) (31.0,155) (32.0,160) (33.0,169) };
\addlegendentry{QVec256};
\addplot[QVec512] coordinates { (14.0,23) (15.0,23) (16.0,23) (17.0,23) (18.0,24) (19.0,24) (20.0,24) (21.0,27) (22.0,37) (23.0,39) (24.0,43) (25.0,44) (26.0,55) (27.0,78) (28.0,113) (29.0,134) (30.0,147) (31.0,160) (32.0,164) (33.0,176) };
\addlegendentry{QVec512};

\legend{};
\end{axis}
\end{tikzpicture}&
  \begin{tikzpicture}[trim axis right]
    \begin{axis}[
      plotLatencySmall,
      title={\textbf{select}},
      xlabel={input size (\(\log n\) symbols)},
      ]
\addplot[pasta_bv] coordinates { (14.0,125.57) (15.0,128.851) (16.0,129.833) (17.0,128.96) (18.0,130.686) (19.0,131.1) (20.0,132.491) (21.0,131.976) (22.0,133.858) (23.0,141.477) (24.0,148.669) (25.0,155.068) (26.0,165.971) (27.0,177.12) (28.0,203.362) (29.0,254.557) (30.0,300.505) (31.0,328.006) (32.0,395.122) (33.0,400.721) };
\addlegendentry{pasta\_bv};
\addplot[sdsl_bv] coordinates { (14.0,60.6929) (15.0,60.4381) (16.0,67.1124) (17.0,63.9424) (18.0,66.0586) (19.0,68.033) (20.0,71.1421) (21.0,73.6309) (22.0,78.1386) (23.0,96.7554) (24.0,113.672) (25.0,126.166) (26.0,136.052) (27.0,153.273) (28.0,187.467) (29.0,242.55) (30.0,304.605) (31.0,386.517) (32.0,428.08) (33.0,453.139) };
\addlegendentry{sdsl\_bv};
      
\addplot[QVec256] coordinates { (14.0,104) (15.0,107) (16.0,107) (17.0,107) (18.0,108) (19.0,108) (20.0,110) (21.0,116) (22.0,123) (23.0,129) (24.0,136) (25.0,147) (26.0,159) (27.0,189) (28.0,230) (29.0,279) (30.0,297) (31.0,367) (32.0,388) (33.0,443) };
\addlegendentry{QVec256};
\addplot[QVec512] coordinates { (14.0,103) (15.0,110) (16.0,110) (17.0,111) (18.0,111) (19.0,111) (20.0,112) (21.0,117) (22.0,126) (23.0,129) (24.0,135) (25.0,139) (26.0,153) (27.0,179) (28.0,221) (29.0,254) (30.0,273) (31.0,315) (32.0,357) (33.0,403) };
\addlegendentry{QVec512};

\legend{};
\end{axis}
\end{tikzpicture}&
  \begin{tikzpicture}[trim axis right]
  \begin{axis}[
    plotLatencySmall,
    title={\textbf{space usage}},
    xlabel={input size (\(\log n\) symbols)},
    ylabel={\begin{tabular}{c}
              additional space (\%)
            \end{tabular}},
          ]
          \addplot[pasta_bv] coordinates { (14.0,1.00098) (15.0,0.524902) (16.0,0.286865) (17.0,0.167847) (18.0,0.108337) (19.0,0.0785828) (20.0,0.0637054) (21.0,0.0562668) (22.0,0.0525475) (23.0,0.0506878) (24.0,0.049758) (25.0,0.049293) (26.0,0.0490606) (27.0,0.0489444) (28.0,0.0488862) (29.0,0.0488572) (30.0,0.0488427) (31.0,0.0488354) (32.0,0.0488318) (33.0,0.0488299) };
          \addlegendentry{pasta\_bv};
          \addplot[sdsl_bv] coordinates { (14.0,6.87866) (15.0,6.43005) (16.0,6.27899) (17.0,19.3382) (18.0,13.0688) (19.0,9.76009) (20.0,7.99971) (21.0,7.07374) (22.0,6.59206) (23.0,6.33787) (24.0,6.20715) (25.0,6.13964) (26.0,6.10643) (27.0,6.09075) (28.0,6.08327) (29.0,6.08178) (30.0,6.08185) (31.0,6.08366) (32.0,6.08596) (33.0,6.08868) };
          \addlegendentry{sdsl\_bv};
          
          \addplot[QVec256] coordinates { (14.0,4.58984) (15.0,3.85742) (16.0,3.515625) (17.0,3.34473) (18.0,3.25928) (19.0,3.21655) (20.0,3.19519) (21.0,3.18451) (22.0,3.17917) (23.0,3.1765) (24.0,3.17516) (25.0,3.1745) (26.0,3.17416) (27.0,3.174) (28.0,3.17391) (29.0,3.17387) (30.0,3.17385) (31.0,3.17384) (32.0,3.17383) (33.0,3.17383) };
          \addlegendentry{QVec256};
          \addplot[QVec512] coordinates { (14.0,3.02734) (15.0,2.29492) (16.0,1.953125) (17.0,1.78223) (18.0,1.69678) (19.0,1.65405) (20.0,1.63269) (21.0,1.62201) (22.0,1.61667) (23.0,1.614) (24.0,1.61266) (25.0,1.612) (26.0,1.61166) (27.0,1.6115) (28.0,1.61141) (29.0,1.61137) (30.0,1.61135) (31.0,1.61134) (32.0,1.61133) (33.0,1.61133) };
          \addlegendentry{QVec512};

          \legend{};
        \end{axis}
      \end{tikzpicture}
\end{tabular}}
  \ref{leg:qvec_latency}
   \caption{Comparison of access, rank, and select query latency for our \emph{quad} vectors with \emph{bit} vectors.\label{fig:bv_comparison}}
   \vspace{-.1cm}
 \end{figure*}
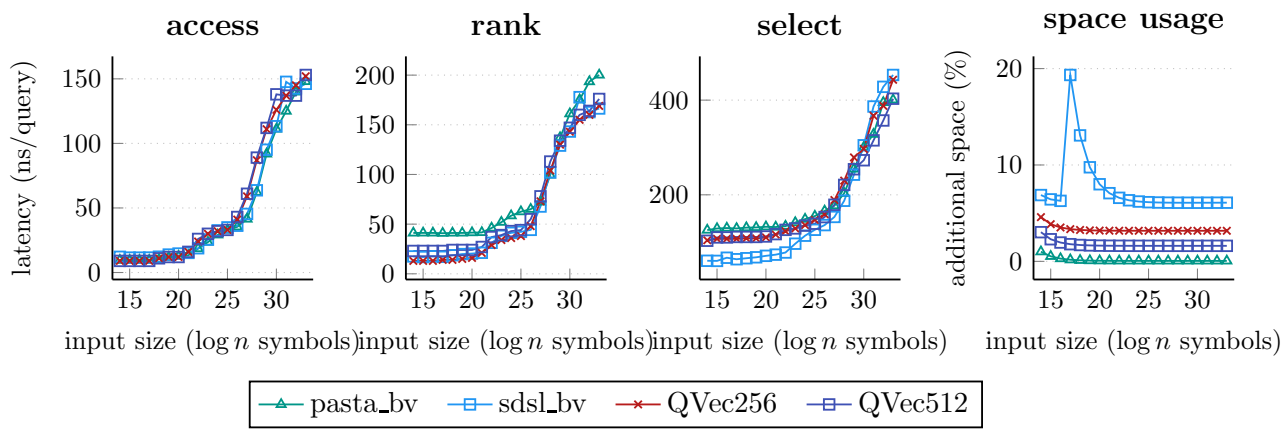
 }

 \longversion{
The running time of access, rank, and select queries of the bit vectors is depicted in \cref{fig:bv_comparison}.
There, we can see that access queries as well as rank and select queries for larger inputs require roughly the same time.
However, for all input sizes, \bvsdsl\ requires roughly the same time (rank) or is faster (select) than \qvector{256}, the quad vector used in the comparison of the wavelet trees.

The additional space usage is as expected and similar to the wavelet trees.
Again, \bvsdsl\ has a spike for smaller inputs, which can be explained by the general overhead of the data structure starting at this input size.

Overall, the evaluation of the bit and quad vectors show that the improvements of three wavelet tree queries considered here are solely due to the algorithmic ideas presented in this paper.
}

\longversion{
\begin{figure*}
  \centering
   \makebox[\textwidth]{
  \input{figures/big_self_comparison.tex}}
  \ref{leg:query_latency}
\caption{Comparison of all our new wavelet trees implementations.
    The first three rows give the access, rank, and select query latency of the implementations, and the last row shows the space overhead of the wavelet tree compared to the input (storing one character per byte).
    The vertical dashed gray lines indicate the L1, L2, L3\(\star\), and L3 cache sizes of the CPU used for these experiments (L3\(\star\) indicates the L3 cache size per core complex).}
  \label{fig:big_wm_self_comparison}
\end{figure*}
}

\longversion{
\section{Conclusion}
We have presented two improvements for wavelet trees that achieve a speedup for access and select queries of up to 2 and for rank queries---which are the most important queries in many applications---up to 3.
To this end, we first changed the underlying tree structure from a binary tree to a 4-ary tree, reducing the number of cache misses in the process.
Furthermore, we introduced the Rank with Additive Approximation problem and showed a small predictive model that solved this problem in practice.
\longversion{This predictive models allows us to prefetch all data necessary for rank queries, resulting in a better speedup compared to access and select queries.}

It remains an open problem to combine the 4-ary wavelet tree layout with the sublinear construction algorithm based on vectorized instructions.
Another interesting line of future research are compressed 4-ary wavelet trees and compressed quad vectors.
\longversion{Additionally, we want to explore bit vectors for even larger alphabets, as our experiments indicate that the techniques proposed in this paper benefit from wavelet trees with more levels.}}

\longandshortversion{
  \bibliography{literature}
}{
  \bibliography{literature_short}
}

\end{document}